\documentclass[a4paper,cleveref,UKenglish,numberwithinsect,thm-restate]{llncs} 

\usepackage[utf8]{inputenc}
\usepackage{tikz}
\usetikzlibrary{arrows,automata}
\usepackage{amsmath}
\usepackage{todonotes}
\usepackage{xspace}
\usepackage{comment}
\usepackage{algorithm}
\usepackage{listings}
\usepackage{cite}
\usepackage{mathpartir}
\usepackage{semantic}

\usepackage{extarrows}
\usepackage[inline]{enumitem}
\usepackage {stmaryrd}

\usepackage{amssymb}
\usepackage{hyperref}
\usepackage{cleveref}
\usepackage{arydshln}
\def\orcidID#1{\smash{\href{http://orcid.org/#1}{\protect\raisebox{-1.25pt}{\protect\includegraphics{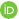}}}}}

\newif\iflong
\longtrue 

\newif\ifdraft
\drafttrue

\renewcommand{\paragraph}[1]{\smallskip\noindent{\textbf{\emph{#1}}}}

\newcommand{\pcmd}{c}

\newcommand{\cmd}{C}
\newcommand{\cmdmap}{\mathcal{C}}

\newcommand{\cmdsteppc}[3]{\xlongrightarrow[]{{#1},{#2},{#3}}}

\newcommand{\En}{\mathit{En}}
\newcommand{\Just}{\mathcal{J}}
\newcommand{\T}{\mathsf{T}}
\newcommand{\run}{\pi}
\newcommand{\tr}{\mu}
\newcommand{\cA}{\mathcal{A}}
\newcommand{\at}[1]{\mathsf{at}\_{#1}}

\newcommand{\inarrT}[1]{\begin{array}[t]{@{}l@{}}#1\end{array}}

\newcommand{\inarr}[1]{\begin{array}{@{}l@{}}#1\end{array}}

\newcommand{\lz}[1]{{\color{green!60!black}{#1}}}
\newcommand{\lo}[1]{{\color{blue!60!black}{#1}}}

\DeclareRobustCommand{\brkstore}{\genfrac[]{0pt}{}}

\newcommand{\imp}{\Rightarrow}
\renewcommand{\implies}{\imp}
\newcommand{\linefill}{\cleaders\hbox{$\smash{\mkern-2mu\mathord-\mkern-2mu}$}\hfill\vphantom{\lower1pt\hbox{$\rightarrow$}}}  
  
\newcommand{\transi}[2]{\mathrel{\lower1pt\hbox{$\mathrel-_{\vphantom{#2}}\mkern-8mu\stackrel{#1}{\linefill_{\vphantom{#2}}}\mkern-11mu\rightarrow_{#2}$}}}

\newcommand{\ntransi}[2]{\mathrel{\lower1pt\hbox{$\mathrel-_{\vphantom{#2}}\mkern-8mu\stackrel{#1}{\linefill_{\vphantom{#2}}}\mkern-8mu\nrightarrow_{#2}$}}}

\newcommand{\textdom}[1]{\ensuremath{\mathsf{#1}}}

\newcommand{\kw}[1]{\textbf{\textdom{#1}}}
\newcommand{\skipc}{\kw{SKIP}}
\newcommand{\ite}[3]{\kw{if}\;#1\:\kw{then}\;#2\; \kw{else}\;#3}

\newcommand{\while}[2]{\kw{while}\;#1\;\kw{do}\;#2}
\newcommand{\sep}{\;\kw{;}\;}

\newcommand{\ALT}{\;\;|\;\;}

\newcommand{\memstate}{{q}}

\newcommand{\epsl}{\varepsilon}
\newcommand{\state}{\kw{q}}

\newcommand{\cs}[1]{\overline{#1}}

\newcommand{\Aux}{{\sf Aux}}

\newcommand{\writeInstn}{\kw{STORE}}

\newcommand{\readInstn}{\kw{LOAD}}

\newcommand{\faddInstn}{\kw{FADD}}

\newcommand{\lQ}{{\mathsf{Q}}}
\newcommand{\linit}{{\mathsf{Q}_0}}
\newcommand{\lTheta}{{\mathbf{\Theta}}}
\newcommand{\M}{\ensuremath{\mathcal{M}}}

\newcommand{\readInst}[2]{#1 \;{:=}\;\readInstn({#2})}
\newcommand{\writeInst}[2]{\writeInstn(#1,#2)}

\newcommand{\assignInst}[2]{#1\;{:=}\;#2}
\newcommand{\faddInst}[3]{#1 \;{:=}\;\faddInstn({#2},{#3})}

\newcommand{\assert}[1]{{\color{blue}{\ensuremath{\left\{
    \begin{array}[c]{@{}l@{}}
      #1
    \end{array}
\right\}}}}}

\newcommand{\vmax}[2]{{#2} \mskip-3mu\uparrow\mskip-2mu {#1}}

\newcommand{\sem}[1]{\llbracket #1 \rrbracket}
\newcommand{\set}[1]{\{{#1}\}}
\newcommand{\tup}[1]{{\langle{#1}\rangle}}
\newcommand{\fv}{\mathit{fv}}

\newcommand{\lR}{{\mathtt{R}}}
\newcommand{\lU}{{\mathtt{RMW}}}

\newcommand{\lTID}{{\mathtt{tid}}}
\newcommand{\lVAL}{{\mathtt{val}}}
\newcommand{\lTS}{{\mathtt{ts}}}
\newcommand{\lFLAG}{{\mathtt{flag}}}
\newcommand{\lLOC}{{\mathtt{loc}}}
\newcommand{\lVIEW}{{\mathtt{view}}}

\crefname{figure}{Fig.}{Figs.}
\Crefname{figure}{Figure}{Figures}

\crefname{theorem}{Theorem}{Theorems}

\crefname{definition}{Def.}{Defs.}
\Crefname{definition}{Definition}{Definitions}

\crefname{example}{Ex.}{Exs.}
\Crefname{example}{Example}{Examples}

\crefname{section}{\S{}}{\S{}}
\crefformat{section}{#2\S{}#1#3}
\Crefformat{section}{#2\S{}#1#3}

\newcommand{\Tid}{\ensuremath{\mathsf{Tid}}}
\newcommand{\Loc}{\ensuremath{\mathsf{Loc}}}
\newcommand{\View}{\ensuremath{\mathsf{View}}}
\newcommand{\Val}{\ensuremath{\mathsf{Val}}}
\newcommand{\Lab}{\ensuremath{\mathsf{Lab}}}

\newcommand{\OLab}{\ensuremath{\mathsf{OLab}}}
\newcommand{\Reg}{\ensuremath{\mathsf{Reg}}}

\newcommand{\btrue}{T}
\newcommand{\bfalse}{F}
\newcommand{\true}{\ensuremath{\mathit{true}}}
\newcommand{\false}{\ensuremath{\mathit{false}}}
\newcommand{\defeq}{\triangleq}
\newcommand{\rlab}[3]{{\lR}^{#1}({#2},{#3})}
\newcommand{\wlab}[3]{{\lW}^{#1}({#2},{#3})}
\newcommand{\rmwlab}[4]{{\lU}^{#1}({#2},{#3},{#4})}

\newcommand{\potstore}{\beta} 

\renewcommand{\L}{L}

\newcommand{\LL}{{\mathcal{{L}}}}
\newcommand{\DD}{{\mathcal{{B}}}} 

\newcommand{\chop}{\mathbin{;}} 

\newcommand{\loc}{{x}}
\newcommand{\loca}{{y}}

\newcommand{\tid}{{\tau}}
\newcommand{\tida}{{\eta}}

\newcommand{\reg}{{a}}

\newcommand{\lab}{{\ell}}
\newcommand{\mlab}{{\iota}} 
\newcommand{\ulab}[4]{{\lU}^{#1}({#2},{#3},{#4})}
\newcommand{\lW}{{\mathtt{W}}}
\newcommand{\valr}{\val_\lR}
\newcommand{\valw}{\val_\lW}

\newcommand{\astep}[1]{\mathrel{\raisebox{-0.8pt}{\ensuremath{\xrightarrow{#1}}}}}

\newcommand{\asteptidlab}[3]{{}\mathrel{\raisebox{-0.8pt}{\ensuremath{\xrightarrow{{#1},{#2}}}}_{#3}}{}}
\newcommand{\asteptidlabpc}[4]{{}\mathrel{\raisebox{-0.8pt}{\ensuremath{\xrightarrow{{#1},{#2},{#3}}}}_{#4}}{}}
\newcommand{\asteplab}[2]{{}\mathrel{\raisebox{-0.8pt}{\ensuremath{\xrightarrow{#1}}}_{#2}}{}}

\newcommand{\msc}{q}

\makeatletter
\newcommand{\vect}[1]{%
  \vbox{\m@th \ialign {##\crcr
  \vectfill\crcr\noalign{\kern-\p@ \nointerlineskip}
  $\hfil\displaystyle{#1}\hfil$\crcr}}}
\def\vectfill{%
  $\m@th\smash-\mkern-7mu%
  \cleaders\hbox{$\mkern-2mu\smash-\mkern-2mu$}\hfill
  \mkern-7mu\raisebox{-3.81pt}[\p@][\p@]{$\mathord\mathchar"017E$}$}
 
\makeatother

\newcommand{\eexp}{E}
\renewcommand{\exp}{e}

\newcommand{\inter}{I}
\newcommand{\sees}[2]{{#1} \!\ltimes\! #2 }
\newcommand{\curr}[1]{\vec{#1}}
\newcommand{\dist}{\mathsf{dist}}

\newcommand{\val}{v}

\newcommand{\makemodel}[1]{\ensuremath{{\mathsf{#1}}}\xspace}

\newcommand{\RA}{\makemodel{RA}}

\newcommand{\TSO}{\makemodel{TSO}}

\newcommand{\SC}{\makemodel{SC}}

\newcommand{\strongC}{\makemodel{StrCOH}}

\newcommand{\pico}{\makemodel{Piccolo}}

\newcommand{\ctid}[1]{\mathtt{T}_#1}
\newcommand{\cloc}[1]{\mathtt{#1}}
\newcommand{\creg}[1]{\mathtt{#1}}
 
\newcommand{\served}{\mathtt{srv}}
\newcommand{\lnext}{\mathtt{nxt}}
\newcommand{\gnext}{\mathtt{next}}
\newcommand{\ego}{\mathsf{ego}}
\newcommand{\tidc}{k}

\newcommand{\glock}{\mathtt{free}}
\newcommand{\llock}{\mathtt{u}}
\newcommand{\gsignal}{\mathtt{sig}}
\newcommand{\lsignal}{\mathtt{s1}}
\newcommand{\lsignalloop}{\mathtt{s2}}

\newcommand{\inv}{\mathcal{I}}

\newcommand{\regstore}{{\gamma}}

\newcommand{\pc}{{l}}
\newcommand{\PC}{\ensuremath{\mathsf{PC}}}
\newcommand{\pcl}{\mathtt{l}}
\newcommand{\pcm}{\mathtt{m}}
\newcommand{\pcmap}{pc}
\newcommand{\npc}{\mathsf{next}}
\newcommand{\npcT}{\mathsf{next^T}}
\newcommand{\npcF}{\mathsf{next^F}}

\newcommand{\cov}[1]{\mathtt{C}(#1)}
\newcommand{\tcovered}{\btrue}
\newcommand{\fcovered}{\bfalse}

\newcommand{\tra}{t}
\newcommand{\vra}{V}
\newcommand{\mra}{m}
\newcommand{\memra}{M}
\newcommand{\tsra}{\nu}
\newcommand{\prop}{\sf prop}

\newcommand{\statera}{q}

\newcommand{\intab}[1]{\begin{tabular}{l}#1\end{tabular}}

\newcommand{\N}{\mathbb{N}}

\newcommand{\map}{map}

\title{Towards Proving Liveness on Weak Memory (Extended Version)\thanks{
Bargmann and Wehrheim are supported by the German Research
Council DFG (project no. 467386514).
}}
\author{Lara Bargmann\orcidID{0009-0004-8778-9098}  \and Heike Wehrheim\orcidID{0000-0002-2385-7512}}
\institute{Carl von Ossietzky Universität Oldenburg, Oldenburg, Germany
\email{lara.bargmann@uol.de, heike.wehrheim@uol.de}
}

\authorrunning{Bargmann and Wehrheim}

\begin{document}
\maketitle

\begin{abstract}
Reasoning about concurrent programs executed on weak memory models is an inherently complex task.
So far, existing proof calculi for weak memory models only cover {\em safety} properties. 
In this paper, we provide the first proof calculus for reasoning about {\em liveness}. Our proof calculus is based on Manna and Pnueli's proof rules for response under weak fairness, formulated in linear temporal logic. Our extension includes the incorporation of {\em memory fairness} into rules as well as the usage of {\em ranking functions} defined over weak memory state. 
We have applied our reasoning technique to the Ticket lock algorithm and have proved it to guarantee starvation freedom under memory models Release-Acquire and StrongCoherence for any number of concurrent threads. 
\end{abstract}

\section{Introduction}
\label{sec:intro}

Verification of concurrent programs is inherently complex due to the fine-grained interleavings of program steps. 
This complexity is further increased when {\em weak memory models} of multiprocessor hardware architectures or modern programming languages need to be taken into account. 
As existing reasoning techniques often assume {\em sequential consistency} (\SC, \cite{DBLP:journals/tc/Lamport79}), they cannot directly be applied in a setting with weak memory. 

Recent years have thus seen numerous approaches transferring verification techniques for concurrent programs to weak memory models.  
These approaches -- automatic and deductive -- almost exclusively concentrate on {\em safety} proofs. Program termination or more generally {\em liveness} has only recently been considered when studying {\em fairness} in weak memory: Both Lahav et al.~\cite{DBLP:journals/pacmpl/LahavNOPV21} and Abdulla et al.~\cite{DBLP:conf/cav/AbdullaAGKV23,DBLP:conf/birthday/AbdullaAGKV24} provide notions of fairness for weak memory that specifically incorporate effects arising in weak memory models. 

Here, we develop the first {\em proof calculus} for reasoning about liveness properties of concurrent programs on weak memory models. More specifically, we start with Manna and Pnueli's (MP's) proof rules for response properties specified in linear temporal logic (LTL). These rules involve a notion of (weak) fairness with respect to program steps. Our first extension of the rules is the incorporation of Lahav et al.'s~\cite{DBLP:journals/pacmpl/LahavNOPV21} notion of {\em memory fairness}. In particular, we allow for giving memory model internal steps as ``helpful transitions" (in Manna and Pnueli terminology) to be treated fairly. 
Our second extension is the usage of {\em ranking functions}  defined over weak memory state. Ranking functions are typically employed in liveness or termination proofs to measure the distance to a target (e.g., end of program). 
 
Instead of defining proof rules for {\em one} fixed memory model, we build on the generic approach of Bargmann et al.~\cite{DBLP:conf/fm/BargmannDW24}, allowing for proofs (potentially) valid in multiple memory models. This approach employs the semantic domain of {\em potentials} for weak memory state, and the logic \pico of Lahav et al.~\cite{DBLP:conf/cav/LahavDW23} to specify assertions over potentials. We extend \pico with further concepts and proof rules to allow reasoning about liveness.  
The thus constructed liveness proofs are then valid for all memory models in which the employed proof rules are sound. To determine rule soundness, the semantic domains of specific memory models need to be mapped onto potentials. 
Here, we provide such mappings and soundness results for the memory models {\em Release-Acquire} (\RA, \cite{DBLP:conf/ecoop/KaiserDDLV17,DBLP:journals/pacmpl/LahavNOPV21}) and Strong Coherence\footnote{Strong Coherence is discussed in the appendix.} (\strongC, \cite{DBLP:journals/pacmpl/LahavNOPV21}).  

Using our \pico extension, the adapted Manna-Pnueli proof rules (in particular rule WELL-JP for {\em parameterized programs}) and our soundness results for proof rules in weak memory models, we prove
starvation freedom of the Ticket lock algorithm~\cite{DBLP:journals/tocs/Mellor-CrummeyS91} on \strongC and \RA for any number of concurrent threads.

\section{Program Syntax and Semantics} 
\label{sec:syntax}
\begin{figure}[t]
  \centering
  \small
  \begin{tabular}{rll@{\ \ \  }rll}
\emph{values} & $\,$ & $\val \in \Val = \set{0,1,\ldots}$ 
&  \emph{locations} & $\,$ & $\loc, \loca \in \Loc = \set{\cloc{x},\cloc{y},\ldots}$
\\ \emph{local registers} & $\,$ & $\reg \in \Reg = \set{\creg{a},\creg{b},\ldots}$ 
&  \emph{thread identifiers} & $\,$ & $\tid \in \Tid = \set{\ctid{0},\ctid{1},\ldots}$
\\ \emph{program counters} & $\,$ & $\pc \in \PC = \set{\pcl_0,\pcl_1, \pcm_0,\ldots}$ 
&  &  & 

\end{tabular}

\[\begin{array}{@{} l l @{}}
\exp ::=  & \reg \ALT \val \ALT \exp + \exp \ALT \exp = \exp \ALT \exp > \exp \ALT \neg \exp \ALT \exp \land \exp \ALT \exp \lor \exp \ALT \ldots
\\[0.5ex]
\pcmd ::=  &
\inarrT{\assignInst{\reg}{\exp} 
\ALT \writeInst{\loc}{\exp}
\ALT \readInst{\reg}{\loc}}  
\ALT \faddInst{\reg}{\loc}{\exp}
 \\[0.5ex]
\cmd ::=  &
\pc: \skipc
\ALT \pc: \pcmd
 \ALT  \cmd \sep \cmd \ALT \pc: \ite{\exp}{\cmd}{\cmd} 
\ALT \pc: \while{\exp}{\cmd} 
\end{array}
\]
\vspace{-10pt}
\caption{Program syntax}
\label{fig:syntax}
\end{figure}

We start by giving the syntax and semantics of concurrent programs on weak memory models. We consider concurrent programs with top-level parallelism only. 
Sequential programs (defined in \cref{fig:syntax}) contain standard programming constructs with named program positions (from a set of program counter values $\PC$, disjoint between threads). 
A mapping $\mathcal{C}$ from thread identifiers to sequential programs describes a concurrent program where $\mathcal{C}(\tid)$ is the sequential program of thread $\tid$ in $\mathcal{C}$.
Threads use disjoint local registers in $\Reg$ and shared locations in $\Loc$. Both take values out of a set $\Val$. In \cref{fig:syntax}, we see statements for storing to and loading from shared locations (write and read instructions), as well as a read-modify-write instruction (RMW) called fetch-and-add (\kw{FADD}).

An example program using this syntax is given in \Cref{fig:waiting}. Here, thread $\ctid{1}$ frees a lock (by writing to $\glock$) and then sends a signal to potentially waiting threads (by writing to $\gsignal$). Thread $\ctid{2}$ is one such thread and loads the values of both shared locations into local registers. If the lock is not free ($\llock\neq 1$), $\ctid{2}$ enters its while loop and waits for the signal. More specifically, thread $\ctid{2}$ waits for the signal {\em to change} by repeatedly reading it and comparing it to $\lsignal$. 

\begin{figure}[t]
  \centering
  \vspace{-10pt}
  \scalebox{0.99}{
  \begin{tabular}[t]{@{}c@{}} 
  $\begin{array}{@{}l@{~}||@{~}l@{}}
    \begin{array}[t]{l}
     \textbf{Thread } \ctid{1}
     \\      
     \pcl_0: \writeInst{\glock}{1};  \\        
     \pcl_1: \writeInst{\gsignal}{1} \\
     \pcl_2:
     \end{array}
    & 
    \begin{array}[t]{l}
     \ \textbf{Thread } \ctid{2}
     \\
      \ \pcm_0: \readInst{\lsignal}{\gsignal}; \\
      \ \pcm_1: \readInst{\llock}{\glock}; \\
      \ \pcm_2: \kw{if}\ (\llock\neq 1) \ \kw{ then }\\
      \qquad \pcm_3: \readInst{\lsignalloop}{\gsignal};\\
      \qquad \pcm_4: \kw{while}\ (\lsignalloop = \lsignal)\ \kw{do} \\        
      \qquad \qquad \pcm_5: \readInst{\lsignalloop}{\gsignal} \\
      \ \pcm_6: 
     \end{array}
     \end{array}$ \\
     ${\color{blue} \Box(\at{\pcm_0} \rightarrow \Diamond \at{\pcm_6})} \qquad \qquad $
  \end{tabular} }
\caption{Concurrent program \texttt{Waiting}; {\color{blue} property} at the end specifying thread $\ctid{2}$ to terminate. 
}
\label{fig:waiting}
\vspace*{-5pt}
\end{figure}

As the semantics of such programs on weak memory models, we give labelled transition systems (LTS) which are constructed in three steps.

\begin{description}[leftmargin=0pt,itemsep=2pt,topsep=2pt]
\item[Local semantics.] We first give a {\em local} semantics (program semantics without the memory model) in   \cref{fig:pcmdsem}. Therein, {\em labels} (of type $\Lab = \{\rlab{}{\loc}{\valr}, \wlab{}{\loc}{\valw}$, $\ulab{}{\loc}{\valr}{\valw}\mid \loc\in\Loc, \valr,\valw\in\Val \}$) for each
  action (read/write/read-modify-write) associated with each command are fixed. Assignments to local registers lead to transitions labelled with $\varepsilon$. The 
  semantics tracks a local register 
  store $\gamma \in \Reg \to \Val$.  Initially, all values of registers are 0.
  Note that in the read rules of \cref{fig:pcmdsem} ($\kw{LOAD}$ and $\kw{FADD}$), the value read is arbitrary.
  In the combined program semantics, this value will be
  fixed by the memory model semantics. 
  
  \Cref{fig:pcmdsem} only details the semantics of simple commands $\pcmd$. The appendix in \Cref{app:opsem} gives the operational semantics of sequential and concurrent programs using the transitions of commands $\pcmd$ ($\pcmd \gg \regstore \astep{\lab_\varepsilon} \regstore'$) as derived in the rules of \cref{fig:pcmdsem} $(\lab_\varepsilon \in \Lab \cup \{\varepsilon\})$.
  Therein, the state of programs also includes a program counter map $\pcmap \in \Tid \to \PC$ (initially being the first program position for every thread). The map $\pcmap$ is used for building assertions over program positions (like $\at{\pcm_0}$ occurring in the property of \cref{fig:waiting}). We assume a function $\npc$  giving us the position of the next program instruction.  Altogether, the local semantics fixes an LTS for a concurrent program $\mathcal{C}_0$ with states of the form $\langle \mathcal{C}, (\gamma,\pcmap)\rangle\in\mathcal{C}_0.\lQ$ (remaining concurrent program $\mathcal{C}$ plus register and program counter map), initial states $\mathcal{C}_0.\lQ_0 = \set{\langle \mathcal{C}_0, (\regstore_0,\pcmap_0)\rangle}$ and the transitions $\rightarrow$ of Appendix~\ref{app:opsem}.

\item[Weak memory semantics.] The semantics of weak memory models is also an LTS, $\M$, with a set of states denoted
  by $\M.\lQ$, initial states $\M.\linit$, and transitions  
  $\mathop{\asteplab{}{\M}}$. Transition labels of $\M$ consist of
  observable program transition labels (elements of $\Tid \times
  (\Lab \cup \{\epsl\})$) 
  and a (disjoint) set $\M.\lTheta$ of unobservable {\em internal} memory labels.  As an
  example, we present the \SC memory model below. 
  The \RA model is presented in Fig~\ref{fig:RA} and the \strongC model in the appendix.  
  \begin{example}[\SC memory model]\label{sec:semantics}
    The memory model $\SC$ simply tracks the most recent value written
    to each variable (plus the id of the writing thread), 
    i.e., $\SC.\lQ \defeq \Loc \to (\Val \times \Tid)$. $\SC$ has no
    internal memory steps (i.e.,
    $\SC.\lTheta \defeq \emptyset$), and the initial state is defined
    by $\SC.\linit \defeq \set{\lambda \loc \ldotp \tup{0,\ctid{0}}}$ (storing initial value plus 
    $\ctid{0}$ as the initializing thread), and transitions
    are 
    given by: 
  \begin{mathpar}
   \small
    \inferrule*[left=write]{
      \lab=\wlab{}{\loc}{\valw} \\
      \msc'=\msc[\loc \mapsto \tup{\valw,\tid}]  
    }{{\msc} \asteptidlab{\tid}{\lab}{\SC} {\msc'} }
    \qquad 
    \inferrule*[left=read]{
      \lab=\rlab{}{\loc}{\valr} \\
      \msc(\loc)=\tup{\valr,\cdot} \\
    }{{\msc} \asteptidlab{\tid}{\lab}{\SC} {\msc} }
  \end{mathpar}
The semantics of the read-modify-write corresponds to a read and a write in one atomic step. 
  \end{example}
  
\item[Combined program semantics.] The two semantics are combined into a joint semantics $\cs{\M}$ using
  the three rules below for steps corresponding to (1) joint executions of reads, writes and RMWs, (2) steps of the program only and (3) internal steps of the memory model $\M$ only: 
  \begin{mathpar}
{\small
\inferrule*[left=(1)]{
\tup{\cmdmap,(\regstore,\pcmap)} \cmdsteppc{\tid}{\lab} {\pc} \tup{\cmdmap',(\regstore',\pcmap')}
\\\\
\lab \in \Lab \\\memstate {\asteptidlab{\tid}{\lab}{\M}} \memstate'
}{\tup{\cmdmap,(\regstore,\pcmap),\memstate} \asteptidlabpc{\tid}{\lab}{\pc}{\cs{\M}} \tup{\cmdmap',(\regstore',\pcmap'),\memstate'}} 
\and 
\inferrule*[left=(2)]{
\tup{\cmdmap,(\regstore,\pcmap)} \cmdsteppc{\tid}{\epsl}{\pc} \tup{\cmdmap',(\regstore',\pcmap')}}
{\tup{\cmdmap,(\regstore,\pcmap),\memstate} \asteptidlabpc{\tid}{\epsl}{\pc}{\cs{\M}} \tup{\cmdmap',(\regstore',\pcmap'),\memstate}} } 
\end{mathpar}
{\small \begin{mathpar} 
\inferrule*[left=(3)]{
\mlab \in \M.\lTheta  \\ 
\memstate \asteplab{\mlab}{\M} \memstate' 
}{\tup{\cmdmap,(\regstore,\pcmap),\memstate} \asteplab{\mlab}{\cs{\M}} \tup{\cmdmap,(\regstore,\pcmap),\memstate'}}
\end{mathpar}}
\end{description}

\medskip 

\noindent In summary, for a given concurrent program $\mathcal{C}_0$ and memory model $\M$, we thus get an LTS $\cs{\M}=(\lQ,\OLab,\lTheta,\linit,\T)$ of states $\lQ=\mathcal{C}_0.\lQ\times\M.\lQ$, observable labels $\OLab = \Tid \times
  (\Lab \cup \{\epsl\}) \times \PC$ and unobservable labels $\lTheta=\M.\lTheta$, 
  initial states $\linit=\mathcal{C}_0.\linit\times\M.\linit$, and transitions $\T= \mathord{\to_{\cs{\M}}}$ (out of $\lQ \times (\OLab \cup \lTheta) \times \lQ$) defined by the rules (1), (2) and (3). 
  
\begin{figure}[t]
\scalebox{0.85}{
\begin{mathpar}
\small
\inferrule*{
\regstore'=\regstore[\reg \mapsto \regstore(\exp)]
}{
\assignInst{\reg}{\exp} \gg \regstore \astep{\varepsilon} \regstore'
}
\quad
\inferrule*{
\lab =\wlab{}{\loc}{\regstore(\exp)} \\
}{
\writeInst{\loc}{\exp} \gg    \regstore \astep{\lab} \regstore
}
\quad
\inferrule*{
\lab =\rlab{}{\loc}{\valr} \\\\
\regstore'=\regstore[\reg \mapsto \valr]
}{
\readInst{\reg}{\loc} \gg   \regstore \astep{\lab} \regstore'
}
\quad 
\inferrule*{
\lab =\ulab{}{\loc}{\valr}{\valr+\regstore(\exp)} \\\\
\regstore'=\regstore[\reg \mapsto \valr]
}{
\faddInst{\reg}{\loc}{\exp} \gg    \regstore \astep{\lab} \regstore'
}
  \\
\end{mathpar}}
\caption{Local semantics of  commands ($\lab \in \Lab$);
missing cases to be found in \cref{app:opsem}.
}
\label{fig:pcmdsem}
\end{figure}

\section{Fairness, Linear Temporal Logic and a Response Rule}
\label{sec:fair}
We employ linear temporal logic (LTL) for the specification of liveness properties. LTL formulae are typically interpreted over infinite sequences of states. Thus, we assume all states in our LTS to have successors, i.e., for every state $\state \in \lQ$ there has to be a transition $\tr \in \T$ with $\tr=(\state,\cdot,\cdot) \in \T$\footnote{When the semantics of programs derives states without successors, we add an idling transition $(\state,\varepsilon,\state)$ to $\T$.}.  

 \textbf{Fairness.} Fairness is defined on runs of LTS'. A {\em run} $\run$ of an LTS $T=(\lQ,\OLab,\lTheta,\linit,\T)$ is an infinite sequence of transitions $(\state_0,\alpha_0,\state_1) (\state_1,\alpha_1, \state_2) \ldots$ such that $\state_0 \in \linit$ and $(\state_i,\alpha_i,\state_{i+1}) \in \T$, $i \geq 0$. 
The {\em computation} $\sigma$ of a run $\run = (\state_0,\alpha_0,\state_1) (\state_1,\alpha_1, \state_2) \ldots$ is the infinite sequence of states $\sigma = \state_0 \state_1 \state_2 \ldots$. 
A transition $\tr = (\state, \alpha, \state')$ is {\em enabled at position $k$} in $\run$ if $\run[k]=(\state,\cdot,\cdot)$, and $\tr$ is {\em continuously enabled from position $k$ on} if $\tr$ is enabled at all positions $i\geq k$. A transition $\tr \in \T$ is {\em taken at position $k$} in a run $\run$ if $\run[k]=\tr$.  

\begin{definition}\label{def:fair}
    Let $T=(\lQ,\OLab,\lTheta,\linit,\T)$ be a labeled transition system,  $\Just \subseteq \T$ a set of transitions (the {\em just}  transitions). A run $\run$ of $T$ is {\em $\Just$-fair} if every transition $\tr \in \Just$ which is continuously enabled from a position $k$ on is taken at some position $i \geq k$.  
\end{definition}

In the sequel, we identify transitions of program commands by their positions $\pc \in \PC$, i.e., we will also say that a command at position $\pc$ is enabled or taken when a transition $\mu = (\cdot,(\tid,\lab,\pc),\cdot)$ for some $\tid$ and $\alpha =(\tid,\lab,\pc) $ is enabled or taken. 
In particular, a command of thread $\tid$ at position $\pc$ is enabled, $\En(\pc)$, if $\pcmap(\tid)=\pc$. 
Note that besides program transitions, we also have {\em internal} transition of the memory model in an LTS. 
These memory internal transitions are transitions with labels $\mlab\in\lTheta$ (and we identify them with their label $\mlab$).
By allowing them to be part of just sets $\Just$, we can express \textit{memory fairness}, i.e., define runs being fair wrt.~internal memory transitions. 

  \textbf{Linear Temporal Logic.} For LTL, we assume to have a logic at hand for specifying {\em atomic propositions} (i.e., predicates) $\varphi$ describing sets of states. In \cref{sec:pico}, we will introduce the logic \pico for this purpose. Its interpretation fixes the meaning of $\state \models \varphi$ for states $\state$. On top of \pico, we build linear temporal logic formulae using operators $\Box$ (always) and $\Diamond$ (eventually). In our proof rules, we only consider formulae of the form $\Box (\varphi \rightarrow \Diamond \psi)$ ({\em response} properties) where $\varphi$ and $\psi$ are state formulae in \pico. 
  
\begin{definition} 
Let $\run$ be a run of an LTS $T$ and $\sigma$ its corresponding computation. 
For LTL formula $\Psi$ and \pico formula $\varphi$, we define $(\sigma,k) \models \Psi$  as follows: 
\begin{itemize}
  \item  $(\sigma,k) \models \varphi$  ($\varphi$ an atomic propositions) iff $\sigma[k] \models \varphi$ (see \cref{def:pico}), 
  \item $(\sigma,k) \models \Box\ \Psi$ iff $(\sigma,i) \models \Psi$ for all $i \geq k$, 
  \item $(\sigma,k) \models \Diamond \Psi$ iff $(\sigma,i) \models \Psi$ for some $i \geq k$.  
\end{itemize}
An LTS $T$ {\em satisfies} $\Psi$, $T \models \Psi$ , if $(\sigma,0) \models \Psi$   for all computations $\sigma$ of runs of $T$; $T$ {\em fairly satisfies} $\Psi$ with respect to some set of just transitions $\Just$ if $(\sigma,0) \models \Psi$ for all computations $\sigma$ of $\Just$-fair runs of $T$. 
\end{definition} 

The property given at the end of the program in \cref{fig:waiting} is a response formula 
specifying that thread $\ctid{2}$ will eventually reach position $\at{\pcm_6}$, i.e., will terminate.

\begin{figure}[t] \scalebox{.99}{
    \begin{mathpar}
    \inference{\textbf{JW1}.\ \varphi \rightarrow \bigvee_{j=0}^{m} \varphi_j \hfill \\  
    \textbf{JW2}.\ \assert{\varphi_i \wedge  \delta_i = z_i} \ \tr  \ \assert{\big(  \bigvee_{j=0}^{m} (\varphi_j \wedge z_i \succ \delta_j)\ \big) \vee (\varphi_i \wedge z_i = \delta_i) }  \\
    \textbf{JW3.}\ \assert{\varphi_i \wedge   \delta_i = z_i} \ \tr_i \ \assert{ \bigvee_{j=0}^{m} (\varphi_j \wedge z_i \succ \delta_j)} \hfill \\
    \textbf{JW4.} \ \varphi_i \rightarrow \En(\tr_i)\hfill }{\Box(\varphi \rightarrow \Diamond \psi)}
    \end{mathpar}}
    \caption{Rule {\textsc WELL-J} with assertions $\varphi, \psi$, $\varphi_0, \ldots, \varphi_m$ ($\varphi_0=\psi$), 
    $\tr \in \T$ any transition, helpful transitions $\tr_1, \ldots, \tr_m \in \Just$, ranking functions $ \delta_0, \ldots, \delta_m$, $z_i$ meta variables for elements in $\cA$ and $i$ ranging over $1, \ldots, m$.}
    \label{fig:well-j}
\end{figure}

  \textbf{A Proof Rule for Response.}
Manna and Pnueli introduce a number of proof rules for response formulae. Here, we employ the rule WELL-J (for well-founded response under justice) in which all premises of the rule are non-temporal (see \cref{fig:well-j}). We give the premises \textbf{JW2} and \textbf{JW3}  in the form of {\em Hoare triples}~\cite{DBLP:journals/cacm/Hoare69}. 
The rule WELL-J builds on a {\em well-founded domain} $(\cA,\succ)$: $\cA$ is a set and $\succ$ is a transitive and irreflexive relation such that no infinitely descending sequences of elements $a_1 \succ a_2 \succ a_3 \succ \ldots$ exist. We often employ well-founded domains which are {\em products} of other well-founded domains, then using a lexicographic ordering. For measuring distance to targets, rule WELL-J uses {\em ranking functions}  $\delta_i: \lQ \rightarrow \cA$. While trying to reach the target, the program might pass through a number of intermediate states, represented by state assertions $\varphi_0, \ldots, \varphi_m$. 
Progress towards the target is achieved by executing {\em helpful transitions} $\tr_1, \ldots, \tr_m$ which need to be in the justice set $\Just$. 
More specifically, rule WELL-J in Fig.~\ref{fig:well-j} has four  premises (for proving $\Box(\varphi \rightarrow \Diamond \psi)$):
\begin{description}
    \item[JW1]  states that the ``start" property $\varphi$ of the response formula implies one of the intermediate assertions $\varphi_j$. 
    \item[JW2] states that for all intermediate assertions $\varphi_i$ and all transitions $\tr \in \T$, the execution of $\tr$ either leads to states satisfying some $\varphi_j$, thereby decreasing the ranking function, or stays inside $\varphi_i$ and keeps the value of $\delta_i$. 
    \item[JW3]  states executions of helpful transitions $\tr_i$   to definitely decrease ranking functions.
    \item[JW4]  requires showing that helpful transition $\tr_i$ is enabled in states satisfying intermediate assertion $\varphi_i$. 
\end{description}

 This proof rule is sound, i.e., if all four premises hold, an LTS fairly satisfies the response property with respect to $\Just$. 

 \begin{lemma}[Manna\&Pnueli~\cite{DBLP:conf/birthday/MannaP10}]
     Rule WELL-J is sound. 
 \end{lemma}

We can employ this rule to prove the response property in program \texttt{Waiting} (Fig.~\ref{fig:waiting}) under \SC, i.e., prove termination of thread $\ctid{2}$. We see that $\ctid{2}$ either terminates when it does not enter the ``then"-branch of the if statement at all ($\llock=1$), or when the while loop terminates ($\lsignal\neq \lsignalloop$). One of these situations will occur in \texttt{Waiting} since $\gsignal\neq 0\Rightarrow \glock=1$ is an invariant. We employ this invariant when defining the intermediate assertions for WELL-J.
More specifically, we need different intermediate assertions for each position of thread $\ctid{2}$ ($\at{\pcm_0},$ $\at{\pcm_1},$ $\at{\pcm_2},$ $\at{\pcm_3},$ $\at{\pcm_{4}},$ $\at{\pcm_{5}},$ and $\at{\pcm_6}$), each with 
their related helpful transition in $\ctid{2}$. For example, the assertion 
$\at{\pcm_1}\land (\lsignal=0\lor \glock=1)$
has the helpful transition $\pcm_1$. 
Additionally, while $\ctid{2}$ is in its while loop ($\at{\pcm_{4,5}}$), we track the progress of $\ctid{1}$ to ensure that we will eventually leave the loop. This leads to a total of nine state assertions $\varphi_i$ with ranking functions $\delta_i: (i, 1-\gsignal,1-\lsignalloop)$, where $i$ tracks the progress of $\ctid{2}$ to enter and leave its loop, $1-\gsignal$ tracks the progress of $\ctid{1}$ writing to $\gsignal$ and $1-\lsignalloop$ tracks whether $\ctid{2}$ has read the necessary value to leave its loop. 
Basically, all transitions are helpful here since they participate in achieving the progress of $\ctid{2}$.

This proof, however, does not hold for programs under weak memory, as threads then might not see the {\em same} values of shared locations like $\glock$ at the same time. This, in particular, means that state assertions like $\gsignal\neq 0\Rightarrow \glock=1$ are not well defined. 
Nevertheless, rule WELL-J is sound in such a setting, as we still simply work on LTS of programs, albeit with ``unusual" state and ``unusual" memory model internal steps. What we need to properly employ this rule in our setting is a way of defining (1) intermediate assertions as well as (2) ranking functions in terms of weak memory state. 

\section{Potentials and the Piccolo Logic} \label{sec:pico}

For sequential consistency, we use predicates on locations and registers for specifying state assertions. 
This is different when we consider weak memory models, in particular when  
constructing proofs that are potentially valid for multiple memory models. 
Here,  we employ 
\pico~\cite{DBLP:conf/cav/LahavDW23,DBLP:conf/fm/BargmannDW24} for this purpose. \pico 
is an interval-based logic for weak
memory formalized using a notion of {\em potentials}.

\noindent {\textit {Notation.}}   For sets $S$ and lists $L$, we write $\#S$ and $\#L$ for the size and length of $S$, $L$, respectively; $L[i]$, $1 \leq i \leq \#L$, is the $i$-th element in $L$; $L_1 \cdot \L_2$ is the concatenation of lists.  

 \textbf{Potential domain.} 
Potentials~\cite{DBLP:conf/cav/LahavDW23,DBLP:journals/toplas/LahavB22} reflect the phenomenon of weak memory models that threads might still see {\em stale} values of shared locations and thus may read from several possible
writes to a location, stale ones, and more recent ones.  

Potentials reflect this by describing {\em views} of threads which -- contrary to \SC\ -- 
differ between threads. 
Weak memory models typically have memory model internal steps (like flushes) which {\em advance} these views. For liveness proofs, we need to be able to reason about such view advancements. 
For example, for termination in program \texttt{Waiting} it is crucial that thread $\ctid{2}$ eventually sees the latest values of $\glock$ and $\gsignal$ which thread $\ctid{1}$ has written. 

Technically, a potential store is a mapping from locations to tuples of values, thread IDs (of the writer), {\em covered flags} (needed for RMWs), plus some auxiliary information
required by specific memory
models. This auxiliary information differs between memory models: \SC
requires no additional auxiliary information,
while \RA  and \strongC keep track of {\em timestamps} (see \cref{sec:lifting} and appendix).

\begin{definition}
\label{def:pot-stores} 
A \emph{potential store} is a function
$\potstore : \Loc \to \Val \times \Tid \times   \{\btrue,\bfalse\} \times \Aux$, where $\Aux$
captures the auxiliary information.
\end{definition}
We use $\potstore(\loc).\lVAL$, $\potstore(\loc).\lTID$, $\potstore(\loc).\lFLAG$ and (potentially) $\potstore(\loc).\lTS$ to retrieve
the value, thread id, covered flag, and auxiliary timestamp of $\potstore(\loc)$,
respectively. 

Potentials are then sequences (lists) of potential stores.

\begin{definition}\label{def:sra-state} 
A \emph{potential} is a non-empty set of lists of potential stores.  We let $\LL$ be the set of all potentials. 
A \emph{potential mapping}  is a partial function $\DD : \Tid \rightarrow \LL$ 
that maps thread identifiers to potentials such that all lists agree on the last store.

\end{definition}

\begin{example} \label{ex:pot} 
  Consider the following two lists of potential stores
   \[ \L_1={ {\brkstore{\glock\mapsto \tup{0,\ctid{0}, \tcovered}}{\gsignal
      \mapsto \tup{0,\ctid{0}, \tcovered}}\cdot \brkstore{\glock\mapsto
      \tup{1,\ctid{1}, \tcovered}}{\gsignal \mapsto \tup{0,\ctid{0},
        \tcovered}}} } \qquad   \L_2={  { \brkstore{\glock\mapsto
      \tup{1,\ctid{1}, \tcovered}}{\gsignal \mapsto \tup{0,\ctid{0},
        \tcovered}}} } \]
  and the program \texttt{Waiting} in Fig.~\ref{fig:waiting}. In the program, the initialization of locations to value 0 is done by a specific thread $\ctid{0}$. After executing
  instruction $\pcl_0$ of $\ctid{1}$, thread $\ctid{2}$ could have the
  potential $\DD(\ctid{2})=\{\L_1\}$. 
  In this, $\ctid{2}$ 
  currently sees the {\em stale} value 0 of $\glock$ and value 0 of $\gsignal$ . The potential, however, also includes the latest values of $\glock$, as represented by the second potential store in the list. 

  In the future (e.g., after some memory model internal steps), 
  $\ctid{2}$'s potential might only include 
  the  list $L_2$ in which it will see $\glock$ to be $1$ while $\gsignal$ is still 0. Note that the covered flag $\tcovered$ is only relevant for programs containing RMWs and can therefore be ignored here. 
\end{example}

\begin{figure}[t] \small
$\begin{array}{@{} l @{\quad }r l l l @{}}
\text{\em extended expr.} &  
& \eexp & ::=  & \exp \ALT \loc  \ALT \curr{\loc} \ALT \eexp + \eexp 
                            \ALT \eexp -  \eexp \ALT \neg \eexp \ALT \eexp \wedge \eexp \ALT \ldots \\ 

\text{\em interval assertions} & 
& \inter & ::= & [\eexp] \ALT \inter \chop \inter\\  
\text{\em assertions} & 
& \varphi
& ::= & \sees{\tid}{\inter} \ALT \cov{\loc} \ALT 
\dist(\tid,\loc)=n \ALT \at\pc \ALT \exp \ALT \varphi \land \varphi \ALT \ldots 
 \end{array}$
 \vspace{-2pt}
 \caption{Assertions of \pico (excerpt)}
\label{fig:logic} 
\end{figure}

\smallskip 
 \textbf{Logic \pico.} The logic \pico (given in \cref{fig:logic}\footnote{We have elided constructs occurring in~\cite{DBLP:conf/fm/BargmannDW24} which are not needed here.}) is interpreted over potentials.  
The key concept of \pico are {\em interval assertions}: 
a potential list fulfils an interval assertion $[\eexp]$ when all elements in the list satisfy
$\eexp$, and a list $\L$ satisfies $[\inter_1] \chop [\inter_2]$
(where $\chop$ is the {\em chop} operator~\cite{DBLP:journals/corr/abs-1207-3816,DBLP:journals/ipl/ChaochenHR91})
if $\L$ can be split into lists $\L_1$ and $\L_2$ such that $\L_1$ satisfies $[\inter_1]$ and $\L_2$ satisfies $[\inter_2]$. This allows us to write assertions on values of shared locations currently seen by a thread as well as seen in the future. 

Here, we build on the version of \pico in~\cite{DBLP:conf/fm/BargmannDW24} and extend it with three constructs: 
(a) 
we introduce the {\em current} (most recent/last) value $\curr{\loc}$ of a location $\loc$,
(b) introduce the distance function $\dist$ to describe the distance in the potential to the  most recently written value of a location, and 
(c) introduce the {\em covered} predicate  $\cov{\loc}$ for reasoning about RMWs. Intuitively, locations are covered when an RMW has to read from the most recent value of a location.

We use $\vmax{\loc}{\tid}$ as a shorthand for $\dist(\tau,\loc) = 0$ to stand for {\em view maximality}, specifying a thread $\tid$ to only see the most recent value of a location $\loc$. 
Note that we can thereby also state  {\em negated}  view-maximality (by writing $\dist(\tau,\loc) = n$ for some $n \neq 0$).  
In prior works, such negations have been disallowed because they impact the {\em stability} of assertions: An assertion is said to be stable if its validity cannot be changed by memory model internal steps. This is of key importance for safety proofs where we only reason about program steps. For liveness, this is different: here, we want to explicitly reason about memory model internal steps and specify that their execution changes properties, namely the distance of a thread towards observing the most recent value of a shared location. 

Next, we define the interpretation of $\pico$ on the domain of potentials.

\begin{definition} \label{def:pico}
Let
$\regstore$ be a register store, $\pcmap$  a program counter map,
$\potstore$ a potential store, 
$\L$ a store list,
and $\DD$ a potential mapping.  
For an expression $\exp$, we let $\sem{\exp}_{\tup{(\regstore,\pcmap),\potstore,\L}} \defeq \regstore(\exp)$,  
$\sem{\loc}_{\tup{(\regstore,\pcmap),\potstore,\L}} \defeq \potstore(\loc).\lVAL$ 
and $\sem{\curr{\loc}}_{\tup{(\regstore,\pcmap),\potstore,\L}} \defeq \L[\#\L](\loc)$\footnote{Note again that in potential domains all lists have to agree on the very last store.}. 
The extension of this notation to any extended expression $\eexp$ is standard.
The validity of state assertions $\varphi$ in $\tup{(\regstore,\pcmap),\DD}$, denoted by $\tup{(\regstore,\pcmap),\DD} \models \varphi$, is defined as follows:
\begin{enumerate}[topsep=2pt]

\item $\tup{(\regstore,\pcmap),\L} \models {[\eexp]}$ if $\sem{\eexp}_{\tup{(\regstore,\pcmap),\potstore,\L}} = \true$ for every $\potstore \in \L$. 

\item $\tup{(\regstore,\pcmap),\L} \models {\inter_1 \chop \inter_2}$ if
     $\tup{(\regstore,\pcmap),\L_1}  \models \inter_1$
          and
       $\tup{(\regstore,\pcmap),\L_2}  \models \inter_2$
           for some (possibly empty)
  $\L_1$ and $\L_2$ such that $\L=\L_1 \cdot \L_2$.

 \item $\tup{(\regstore,\pcmap),\DD} \models \sees{\tid}{\inter}$ if 
   $\tup{(\regstore,\pcmap),\L} \models \inter$ for every $\L \in \DD(\tid)$. 

    \item $\tup{(\regstore,\pcmap), \DD} \models \cov{\loc}$ if 
          $\L[i](\loc).\lFLAG=\tcovered$ for every $\L \in \DD(\tid), \tid\in\Tid, 1 \leq i \leq \# \L$.

 \item $\tup{(\regstore,\pcmap), \DD} \models\dist(\tid,\loc)=n$ if $\max \{k \mid k = \# \{ \L[i](\loc) \mid \L[i](\loc)\neq\L[\#\L](\loc) \},  \L \in \DD(\tid) \}=n$.

 \item $\tup{(\regstore,\pcmap),\DD} \models \at\pc$ if $\pcmap(\tid)=\pc$ for some $\tid\in\Tid$\footnote{Note that program positions are disjoint between threads so that there is at most one such thread here.}, and $\tup{(\regstore,\pcmap),\DD} \models \exp$ if $\regstore(\exp)=\true$. 
 
   \item $\tup{(\regstore,\pcmap), \DD} \models \varphi_1 \land \varphi_2$ if $\tup{(\regstore,\pcmap), \DD} \models \varphi_1$ and $\tup{(\regstore,\pcmap), \DD} \models \varphi_2$. 
  \end{enumerate}
\end{definition}

\begin{example}
    In \cref{ex:pot} we had the following list of potential stores of thread $\ctid{2}$: 
    \[\L_1={ {\brkstore{\glock\mapsto \tup{0,\ctid{0}, \tcovered}}{\gsignal
      \mapsto \tup{0,\ctid{0}, \tcovered}}\cdot \brkstore{\glock\mapsto
      \tup{1,\ctid{1}, \tcovered}}{\gsignal \mapsto \tup{0,\ctid{0},
        \tcovered}}} }\]
     Assuming this to be the only list in $\DD(\ctid{2})$, this state satisfies $\dist(\ctid{2},\glock)=1$,
     $\sees{\ctid{2}}{[\glock=0] \chop [\glock=1]}$, $\sees{\ctid{2}}{[\gsignal=0]}$,
      $\neg (\vmax{\glock}{\ctid{2}})$ and $\cov{\glock}$. The value of $\overrightarrow{\glock}$ is 1. 
\end{example}

  \textbf{Proof Rules.} 
\pico already comes with a set of base proof rules for Hoare triples~\cite{DBLP:journals/cacm/Hoare69}. So far, these have been used to build Owicki-Gries-like proof outlines~\cite{DBLP:journals/acta/OwickiG76} for safety proofs. Here, we employ the Hoare triples to show (part of) the premises of rule WELL-J. Fig.~\ref{fig:rules} gives some proof rules needed for our examples; further rules, in particular about interval formulae, can be found in~\cite{DBLP:conf/fm/BargmannDW24}. We use the following notation: 
For an assertion $\varphi$, we let $\fv(\varphi) \subseteq \Reg \cup \Loc$ be the set of registers and 
locations occurring in $\varphi$, and we use $\psi \notin \varphi$ to say that formula $\psi$ is not occurring as a subterm in $\varphi$. We will shortly discuss all the rules.

\begin{figure}[t]
\centering
\begin{small}
\scalebox{0.87}{
\begin{tabular}[t]{|l|r@{}c@{}l|l@{}|}
\hline
Assumption & Pre & Command & Post & Reference  \\
 \hline
 
$\loc \notin \fv(\varphi), \at{\pc}\notin\varphi$ & $\assert{\varphi}$ & $\tid \mapsto \pc: \writeInst{\loc}{\exp}$ & $\assert{\varphi}$ & {\sc Stbl-St} \\

\hdashline

\parbox{2.7cm}{$\inarr{\loc, \reg \notin \fv(\varphi), }$ \\
  $\at{\pc}\notin\varphi$,\\
  $\forall\loca:\neg (\vmax{\loca}{\tid}) \notin \varphi$}
  & $\assert{\varphi}$ & $\tid \mapsto \pc: \faddInst{\reg}{\loc}{\exp}$ & $\assert{\varphi}$ & {\sc Stbl-Rmw}\,  \\

  \hdashline

\parbox{2.7cm}{\vspace{0.2em}$\reg \notin \fv(\varphi), \at{\pc}\notin\varphi,$\\$\forall\loca:\neg (\vmax{\loca}{\tid}) \notin \varphi$} & $\assert{\varphi}$ & $\tid \mapsto \pc: \readInst{\reg}{\loc}$
& $\assert{\varphi}$ & {\sc Stbl-Ld} \\

\hline 
 $\at{\pc}\notin\varphi$& $\assert{\varphi(\reg:=\exp)}$ & $\tid \mapsto \pc: \reg:=\exp$ & $\assert{\varphi}$ & {\sc Sub} \\
 & $\assert{\at{\pc}}$ & $\tid \mapsto \pc: \cmd$ & $\assert{\at{\npc(\pc)}}$ & {\sc Pc} \\
 
\hline

 & $\assert{\vmax{\loc}{\tid}}$ & $\tid \mapsto \pc: \writeInst{\loc}{\exp}$ & $\assert{\vmax{\loc}{\tid} \land \sees{\tid}{[\loc=\exp]} \land\\ \curr{\loc}=e }$ & {\sc St} \\

 \hdashline

 \intab{$\tida \neq \tid$, $\loc \notin \fv(\inter_\tid)$} & $\assert{ \sees{\tid}{\inter_\tid}   \wedge  \sees{\tida}{\inter \chop \inter_\tid}}$ & $\tid \mapsto \pc: \writeInst{\loc}{\exp}$ & $\assert{\sees{\tida}{\inter \chop \inter_\tid}}$ & {\sc St-Other}\\

 \hdashline

 & $\assert{\cov{\loc} \land\\ \curr{x} = \exp_1}$ & $\tid \mapsto \pc: \faddInst{\reg}{\loc}{\exp_2}$ & $\assert{\cov{\loc}\land \curr{\loc}=\exp_1+\exp_2 \land \\\sees{\tid}{[\loc=\exp_1 + \exp_2]}  }$ & {\sc Rmw} \\

\hline

 & $\assert{\vmax{\loc}{\tid} \land \exp(\reg:=\curr{\loc})}$ & $\tid \mapsto \pc: \readInst{\reg}{\loc}$ & $\assert{\exp}$ & {\sc Ld} \\
 
 \hdashline

 P & $\assert{\sees{\tid}{[\exp(\reg:=\loc)];\inter}}$ & $\tid \mapsto \pc: \readInst{\reg}{\loc}$ & $\assert{\exp  \vee \psi}$ & {\sc Ld-Sh}\\

\hline
\end{tabular}}
\end{small}
\caption{\pico proof rules ($\varphi(\reg:=\exp)$ means replacement of $\reg$ by $\exp$ in $\varphi$ and $\tid$ is the executing thread and
$P \defeq \begin{array}[c]{@{}c@{}}
      \assert{\sees{\tid}{\inter}}\tid \mapsto 
      \readInst{\reg}{\loc}
      \assert{\psi}
    \end{array}$, $P$ is thus a Hoare triple)
}
\label{fig:rules}
\end{figure}

\begin{itemize} 
\item The rules for stability,  {\sc Stbl-St}, {\sc Stbl-Rmw}, and {\sc Stbl-Ld}, state that \pico formulae not referring to specific re\-gisters and locations, program counters, or negated view maximality, are not affected by load and store instructions.

\item Rule {\sc Sub} states the standard axiom of assignment of Hoare logic, which is here only defined with respect to registers and local expressions. 
 Rule {\sc Pc} is solely used to reason about program counter values.   

\item Rule {\sc St} describes the changes a store has on the potential of the writing thread, namely, if a thread is view-maximal, the only value it can see for location $\loc$ after the store is its own value. Furthermore, this value is the most recent value for $\loc$. 
Rule {\sc Rmw} describes a similar effect but requires the location to be covered, $\cov{\loc}$, so that the RMW cannot read from a stale value. This means in particular, that locations on which only RMWs are executed are always covered.

 \item Finally, the rules {\sc Ld} and {\sc Ld-Sh} describe the loading of values of shared locations into registers, where \textsc{Ld} covers the case in which the reading thread is view maximal. 

 \end{itemize}

\noindent Note that all of these proof rules are entirely syntactic: we have no semantics of store and load instructions on the level of potentials. Hence, such proof rules cannot be proven sound on this level. However, when deriving potentials from weak memory states (as we do in \cref{sec:lifting}), we can show soundness of rules within specific memory models.

\section{Liveness Proofs}
\label{sec:wait}

Next, we apply the proof rule WELL-J together with \pico and its proof rules to two example programs.  
We in particular illustrate how \pico allows us to write ranking functions over weak memory 
state and why we need memory internal steps as helpful transitions. 
\subsection{Program \texttt{Waiting}}

In Fig.~\ref{fig:proof}, we see an overview of the state assertions $\varphi_9,\ldots,\varphi_0$ and their helpful transitions $\mu_i$ necessary for applying the WELL-J rule to our \texttt{Waiting} program. 
\begin{figure}[t]
    \centering
    $\begin{array}{llll}
  \varphi_9: &\at{\pcm_0}\land \inv  & \tr_9: &\pcm_0 \\
  \varphi_8: &\at{\pcm_1}\land(\lsignal=0\lor\sees{\ctid{2}}{[\glock=1]})\land \inv  & \tr_8: &\pcm_1 \\
  \varphi_7: &\at{\pcm_2}\land(\lsignal=0\lor\llock=1)\land \inv  & \tr_7: &\pcm_2 \\
  \varphi_6: &\at{\pcm_3}\land\lsignal=0\land \inv  & \tr_6: &\pcm_3 \\
  \varphi_5: &\at{\pcm_{4,5}}\land\at{\pcl_0}\land\lsignal=0\land \inv  & \tr_5: &\pcl_0 \\
  \varphi_{4}: &\at{\pcm_{4,5}}\land\at{\pcl_1}\land\lsignal=0\land \inv  & \tr_4: &\pcl_1 \\
  \varphi_{3}: &\at{\pcm_{4,5}}\land\at{\pcl_2}\land\lsignal=0\land\neg\vmax{\gsignal}{\ctid{2}}\land \sees{\ctid{1}}{[\gsignal=1]}\land \inv \quad & \tr_3: &\mlab \\
  \varphi_{2}: &\at{\pcm_{4}}\land\at{\pcl_2}\land\lsignal=0\land\vmax{\gsignal}{\ctid{2}}\land \sees{\ctid{1}}{[\gsignal=1]}\land \inv  & \tr_2: &\pcm_4 \\
  \varphi_{1}: &\at{\pcm_{5}}\land\at{\pcl_2}\land\lsignal=0\land\vmax{\gsignal}{\ctid{2}}\land \sees{\ctid{1}}{[\gsignal=1]}\land \inv  & \tr_1: &\pcm_5 \\
  \varphi_0: & \at{\pcm_6}
\end{array}$
    \caption{State assertions and helpful transitions for program \texttt{Waiting}}
    \label{fig:proof}
\end{figure}
During all executions of \texttt{Waiting} the following invariant $\inv$ holds:
\[\inv: \sees{\ctid{2}}{[\gsignal=0];[\glock=1]}\land \sees{\ctid{2}}{[\gsignal=0];[\gsignal=1]}\land \vmax{\gsignal}{\ctid{1}}\]
We need this invariant (more specifically $\sees{\ctid{2}}{[\gsignal=0];[\glock=1]}$) for example to show that transition $\pcm_0$ brings us from $\varphi_9$ to $\varphi_8$. For this, we use the rules \textsc{Pc} and \textsc{Ld-Sh}. $\sees{\ctid{2}}{[\gsignal=0];[\glock=1]}$ holds throughout the entire program because of the \textsc{Stbl} rules and \textsc{St-Other}.

Note that this proof has an additional assertion ($\varphi_3$) compared to the \SC proof in \cref{sec:fair}, and assertions regarding locations are now written in \pico. After $\ctid{1}$ has finished, this is necessary to differentiate between the states where $\ctid{2}$ can only read the newest value 1 of $\gsignal$ and where it could still read 0.  
As helpful transitions of $\varphi_3$, we need memory internal transitions $\mlab$, since in this state, only view advancement will bring the program nearer to termination. 
This is possible in our setting since (memory) fairness allows us to include internal transitions $\mlab$ into the just set $\Just$.

 To apply rule WELL-J, we also need ranking functions for every state assertion. In this proof, they use the well-founded domain $\mathbb{N} \times \mathbb{N} \times\mathbb{N} \times \mathbb{N}$ and are defined as 
\[ \delta_i: (i, 1-\overrightarrow{\gsignal},\dist(\ctid{2},\gsignal),1-\lsignalloop)  \]
for $0\leq i\leq 9$.
In contrast to the \SC proof, our ranking functions here need to include $\dist(\ctid{2},\gsignal)$, measuring the distance of $\ctid{2}$ towards seeing the latest value of $\gsignal$. This helps us to ensure \textbf{JW3} holds for $\varphi_3$ (we still make progress with its helpful transition) even if $\ctid{2}$ is not directly view maximal after executing $\mlab$. 

Next, we will look at a further case study showing starvation freedom of the Ticket lock algorithm, using 
a variant of WELL-J for parameterized programs. 
\subsection{Ticket Lock} 

\begin{figure}[t]
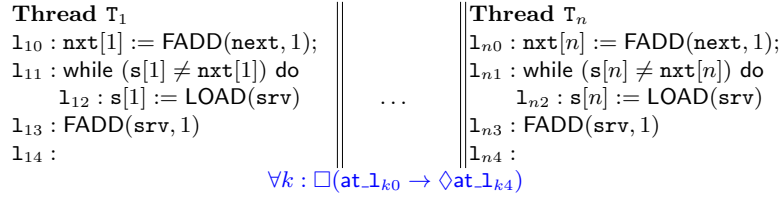

  \centering
  \scalebox{0.95}{
  \begin{tabular}[t]{@{}c@{}} 
  $\begin{array}{@{}l@{~}||@{~}l@{~}||@{~}l@{}}
    \begin{array}[t]{l}
     \textbf{Thread } \ctid{1}
     \\      
     \pcl_{10}: \faddInst{\lnext[1]}{\gnext}{1};  \\        
     \pcl_{11}: \kw{while}\ (\creg{s}[1]\neq\lnext[1])\ \kw{do} \\
     \qquad \pcl_{12}: \readInst{\creg{s}[1]}{\served} \\
     \pcl_{13}: \kw{FADD}(\served,1) \\
     \pcl_{14}: 
     \end{array}
    & \begin{array}[t]{l} \qquad \\ \\ \\ \phantom{bla} \ldots \phantom{bla} \\   \qquad \end{array}& 
    \begin{array}[t]{l}
     \textbf{Thread } \ctid{n}
     \\
       \pcl_{n0}: \faddInst{\lnext[n]}{\gnext}{1};  \\        
     \pcl_{n1}: \kw{while}\ (\creg{s}[n]\neq\lnext[n])\ \kw{do} \\
     \qquad \pcl_{n2}: \readInst{\creg{s}[n]}{\served} \\
     \pcl_{n3}: \kw{FADD}(\served,1) \\
     \pcl_{n4}: 
     \end{array}
     \end{array}$ \\
     ${\color{blue} \forall \tidc: \Box(\at{\pcl_{\tidc 0}} \rightarrow \Diamond \at{\pcl_{\tidc 4})}}$
  \end{tabular} }
    \caption{The Ticket Lock algorithm $\mathtt{Ticket}$ for $n$ threads; property stating starvation freedom for all threads $\ctid{\tidc}$, $\tidc =1, \ldots, n$}
    \label{fig:ticket}
\end{figure}

The Ticket lock algorithm~\cite{DBLP:journals/tocs/Mellor-CrummeyS91} given in Fig.~\ref{fig:ticket} is a synchronization mechanism employed to achieve mutual exclusion. In ticket locks, threads take tickets in the order of their arrival (stored in the register $\lnext[\tidc]$ for thread $\ctid{\tidc}$). 
The thread currently served is the one whose $\lnext$ is equal to the shared variable $\served$. The ``critical section" for which mutual exclusion should be achieved is the position $\at{\pcl_{\tidc 3}}$; after that $\served$ is incremented and the next thread is served. The property we are interested in here is {\em starvation freedom}: every thread wanting to be served is eventually served, i.e., reaches its (empty) critical section.

Here, we assume to have a finite number $n$ of threads running the ticket lock. We thus employ a version of WELL-J for parameterized programs (called WELL-JP), i.e., concurrent programs in which all threads execute the same code\footnote{In our example, the threads' code differs only in the (indices of) names of local registers. }. In this rule, intermediate assertions and helpful transitions can be parameterized (with parameters ranging over thread identifiers). 
\Cref{fig:well-jp} gives the rule WELL-JP. 

\begin{figure}[t]\scalebox{.9}{
    \begin{mathpar}
    \inference{\text{JW1}.\ \varphi \rightarrow  \bigvee_{j=0}^{m} \bigvee_{\tid \in \Tid} \varphi_j[\tid] \hfill \\  
    \text{JW2}.\ \assert{\varphi_i[\tid] \wedge    \delta_i[\tid] = z_i} \ \tr  \ \assert{\big(  \bigvee_{j=0}^{m} (\bigvee_{\tid'  \in \Tid}\varphi_j[\tid'] \wedge z_i \succ \delta_j[\tid'])\ \big) \vee (\varphi_i[\tid] \wedge z_i = \delta_i[\tid]) }  \\
    \text{JW3.}\ \assert{\varphi_i[\tid] \wedge   \delta_i[\tid] = z_i} \ \tr_i[\tid] \ \assert{ \bigvee_{j=0}^{m} (\bigvee_{\tid'  \in \Tid}\varphi_j[\tid'] \wedge z_i \succ \delta_j[\tid'])} \hfill \\
    \text{JW4.} \ \varphi_i[\tid] \rightarrow \En(\tr_i[\tid])\hfill }{\Box(\varphi \rightarrow \Diamond \psi)}
    \end{mathpar}}
    \caption{Rule {\textsc WELL-JP} with assertions $\varphi, \psi, \varphi_0, \varphi_1[\tid], \ldots, \varphi_m[\tid]$ ($\varphi_0=\psi$), 
    $\tr \in \T$ any transition, helpful transitions $\tr_1[\tid], \ldots, \tr_m[\tid] \in \Just$, ranking functions $ \delta_0, \ldots, \delta_m$, $z_i$ meta variables for elements in $A$ and $i=1, \ldots, m$}
    \label{fig:well-jp}
\end{figure}

Intermediate assertions and helping transitions for \texttt{Ticket} are as follows. We let $\mathit{hT}(\tidc) \equiv \at{\pcl_{\tidc 1},\pcl_{\tidc 2},\pcl_{\tidc 3}}$ 
to describe program positions in which $\ctid{\tidc}$ \underline{h}as taken a \underline{T}icket.
We prove the property of starvation freedom for all threads by showing it for some arbitrarily chosen thread called $\ctid{\ego}$. \Cref{fig:ticket-proof} gives intermediate assertions and helpful transitions. 

\begin{figure}[t]
$\begin{array}{llll}
  \varphi_5: &\at{\pcl_{\ego 0}}  & \tr_5: &\pcl_{\ego 0} \\
  \varphi_{4}[\ctid{\tidc}]: &\at{\pcl_{\tidc 1},\pcl_{\tidc 2}} \wedge \mathit{hT}(\ego) \wedge \lnext[\tidc]=\overrightarrow{\served} \wedge \neg (\vmax{\served}{\ctid{\tidc}}) \land  & \tr_{4}[\ctid{\tidc}]: & \mlab\\
  &\sees{\ctid{\tidc}}{[\served = 0]\chop \ldots\chop[\served=\overrightarrow{\served}]} \land \mathtt{s}[\tidc]\leq \overrightarrow{\served}&&\\
  \varphi_{3}[\ctid{\tidc}]: &\at{\pcl_{\tidc 1}} \wedge \mathit{hT}(\ego) \wedge \lnext[\tidc]=\overrightarrow{\served} \wedge \vmax{\served}{\ctid{\tidc}} \land \mathtt{s}[\tidc]\leq \overrightarrow{\served} & \tr_{3}[\ctid{\tidc}]: &\pcl_{\tidc 1} \\
  \varphi_{2}[\ctid{\tidc}]: &\at{\pcl_{\tidc 2}} \wedge \mathit{hT}(\ego) \wedge \lnext[\tidc]=\overrightarrow{\served} \wedge \vmax{\served}{\ctid{\tidc}} \land \mathtt{s}[\tidc]< \overrightarrow{\served} & \tr_{2}[\ctid{\tidc}]: &\pcl_{\tidc 2} \\
  \varphi_{1}[\ctid{\tidc}]: &\at{\pcl_{\tidc 3}} \wedge \mathit{hT}(\ego) \wedge \lnext[\tidc]=\overrightarrow{\served} \wedge \vmax{\served}{\ctid{\tidc}} \land \mathtt{s}[\tidc] = \lnext[\tidc] \quad \ & \tr_{1}[\ctid{\tidc}]: &\pcl_{\tidc 3} \\
  \varphi_0: & \at{\pcl_{\ego 4}}
\end{array}$
\caption{State assertions and helpful transitions for program \texttt{Ticket}}
\label{fig:ticket-proof}
\end{figure}

The ranking functions use the well-founded domain $(\mathbb{N} \cup \{\infty\}) \times (\mathbb{N} \cup \{\infty\}) \times \{0,1,2,3,4,5\} \times (\mathbb{N}\cup \{\infty\})$ and are defined as 
$\delta_5: (\infty,\infty,5,\infty)$, $\delta_0: (0,0,0,0) $
and 
$\delta_i[\ctid{\tidc}]: (\lnext[\ego] - \overrightarrow{\served},\lnext[\tidc]-\mathtt{s}[\tidc],i,\dist(\ctid{\tidc},\served))$,
for $i\in\{1,2,3,4\}$.

The assertions and ranking functions basically state that some thread has a ticket in $\lnext$ which is equal to the most recent value of $\served$, and this thread makes progress towards observing $\overrightarrow{\served}$ (via some helpful transition $\mlab$) as well as towards reading this value and checking it to be the same as its $\lnext$.
By executing $\pcl_{\tidc 3}$,  this thread will then increment $\served$, thus decreasing the distance for $\ego$ to have its own ticket in $\lnext[\ego]$ become the same as $\overrightarrow{\served}$. To be able to show all the premises, we need some further invariants about program $\mathtt{Ticket}$, e.g., $0 \leq \overrightarrow{\served} \leq \overrightarrow{\gnext}$ 
and in particular $\cov{\gnext}$. 

Note that we need the read-modify-write instruction \kw{FADD} in both positions $\pcl_{\tidc 0}$ and $\pcl_{\tidc 3}$. Without them (i.e., with a plain \kw{LOAD} followed by a \kw{STORE}), the executing thread might read a stale value for $\gnext$ (so that multiple threads have the same $\lnext$ then) or for $\served$ (so that the most recent value for $\served$ is not incremented). Fairness wrt.~memory model internal steps does not help in that case as the instructions are not repeatedly executed.

\section{Lifting Memory Models to \pico Proofs } \label{sec:lifting}

Our liveness proof of \texttt{Waiting} is so far completely independent of 
a concrete memory model. 
We, however, ultimately want to be able to say whether this proof is valid for a specific memory model. 
To this end, we need to fix how we interpret \pico on a memory model $\M$. 
\begin{definition}\label{def:picomap}

A function $\map_\M:\M.\lQ\to(\Tid\to\LL)$ is a {\em  lifting function} of \M\ if $\map_\M$ fulfills the following property for all initial states $q\in\M.\linit$:
\begin{align*}
    \tag{\textsc{Init}}
    \tup{(\regstore,\pcmap),\map_\M(q)}\models\forall\tid\in\Tid\ \forall\loc\in\Loc:\vmax{\loc}{\tid}\land\cov{\loc}\land\sees{\tid}{[\loc=0]}
\end{align*}
\end{definition}

A lifting function thus maps memory states to potential mappings, and in particular maps initial states of a memory model to ``initial" states in the potential domain.  Such lifting functions can be employed for interpreting \pico formulae on states $\state=\tup{\mathcal{C},(\regstore,\pcmap),q}$ of the combined semantics 
by defining 
$\state\models\varphi$ to be $\tup{(\regstore,\pcmap),\map_\M(q)}\models\varphi$.

 \begin{figure}[t]
\centering
\begin{small}
\scalebox{0.87}{
\begin{tabular}[t]{|l|r@{}c@{}l|l@{}|}
\hline
\scalebox{0.87}{Assumption} & Pre &\ Command&\  Post & Reference\,  \\
 \hline

\parbox{2.5cm}{
$\forall\tid\forall\loc\forall n>0:$ \\$\dist(\tid,\loc)=n \notin \varphi$
}\, 
& $\assert{\varphi}$ & $\mlab$ 
& $\assert{\varphi}$ & {\sc Stbl-Int} \\

\hdashline 

 & $\assert{\dist(\tid,\loc)=n\land n>0}$ & $\mlab$
& $\assert{\dist(\tid,\loc)= m \land m<n}$ & {\sc Dist} \\

\hline
\end{tabular}}
\end{small}
\caption{\pico proof rules for memory internal transitions labelled $\mlab\in\M.\lTheta$ 
}
\label{fig:intrules}
\end{figure}

With this in mind,  we define what it means for a memory model to fulfill a \pico proof rule. 
\begin{definition}\label{def:soundness}
    The \pico proof rule $\{\varphi\}\ \tid \mapsto \pc: \pcmd\ \{\psi\}$  is {\em sound} in a memory model $\M$ if for every LTS $\cs{\M}=(\lQ,\OLab,\lTheta,\linit,\T)$ we have 
    $$\forall\state,\state'\in\lQ: \state\models\varphi\land \state\astep{\tid, lab(\pcmd), \pc}\state' \Rightarrow \state'\models\psi$$
     where $lab(\pcmd)$ maps each command $\pcmd$ to its respective labels in $\Lab_\varepsilon$.  Similarly, we define $\{\varphi\}\ \mlab \ \{\psi\}$ for $\mlab \in \lTheta$. 
\end{definition}

For achieving validity of liveness proofs for memory models, the following basic healthiness conditions furthermore need to be satisfied by a memory model. 

 \begin{definition}\label{def:picointernal} 
     We say that a lifting function $\map_\M$ of a memory model $\M$ guarantees {\em view advancement of internal transitions}, if the following conditions hold for all internal labels $\mlab\in\M.\lTheta$ and states $\memstate\in\M.\lQ$: 
     \begin{enumerate}
         \item $ 
      \tup{(\regstore,\pcmap),\map_\M(\memstate)}\models(\exists \tid\exists\loc: \dist(\tid,\loc)\neq 0)\implies En(\mlab)$, and 
        \item the rules in \cref{fig:intrules} are sound for $\mlab$-labelled transitions. 
     \end{enumerate}
 \end{definition}

 Recall $\dist(\tid,\loc)\neq 0$ implies that thread $\tid$ cannot see the newest value for $\loc$ at the moment (its distance to this value is not zero).  
 For liveness proofs via rule WELL-J, we want a helpful transition to decrease this distance such that we eventually reach a state where $\dist(\tid,\loc)= 0$ and $\tid$ therefore sees the newest value of $\loc$. \Cref{def:picointernal} guarantees that we can use memory internal transitions as helpful transitions, since a $\mlab$-labelled transition is enabled when the distance is not zero (condition 1) and it decreases $\dist$ (condition 2).

\begin{figure}[t]

\begin{mathpar}
 \inferrule*{ \lab=\rlab{}{\loc}{\valr} \\ \tup{\loc : \valr @ \tra, \tid, \vra}\in\statera.\memra\\ 
 \statera.\tsra(\tid)(\loc)= \tra
   }
   {\statera\asteptidlab{\tid}{\lab}{\RA} \statera } 
   \and 
   \inferrule*{\lab=\wlab{}{\loc}{\valw} \\   \neg \exists \mra\in\statera.\memra_{|\loc}: \mra.\lTS =\tra \\
   \statera.\tsra(\tid)(\loc) <\tra
   }
   {\statera\asteptidlab{\tid}{\lab}{\RA} \tup{\lz{\statera.\memra\cup \{\tup{\loc :\valw @\tra, \tid,\statera.\tsra(\tid) [\loc\mapsto\tra]]}}\} ,\lo{\statera.\tsra [\tid \mapsto \statera.\tsra(\tid) [\loc\mapsto\tra]]}} } 
  \and 
   \inferrule*{ \lab = \rmwlab{}{\loc}{\valr}{\valw} \\ \lab_1=\rlab{}{\loc}{\valr} \\ \lab_2=\wlab{}{\loc}{\valw}\\\\
   \statera'.\tsra(\tid)(\loc) =\statera.\tsra(\tid)(\loc)+1\\
   \statera\asteptidlab{\tid}{\lab_1}{\RA} \statera\asteptidlab{\tid}{\lab_2}{\RA} \statera'   }
 {\statera\asteptidlab{\tid}{\lab}{\RA} \statera' } 
 \and 
   \inferrule*{\lab=\prop \\   \mra\in\statera.\memra \\  \exists\tid\in\Tid:\tsra(\tid)(\mra.\lLOC)<\mra.\lTS 
   }
   {\statera\asteplab{\lab}{\RA} \tup{\lz{\statera.\memra},\lo{\statera.\tsra [\tid \mapsto \statera.\tsra(\tid) \sqcup \mra.\lVIEW}]  } } 
\end{mathpar}
\vspace{-10pt}\caption{Operational semantics of \RA using colors to highlight the updated \lz{memory} and \lo{view} components; 
$\memra_{|\loc}\subseteq\memra$ is the set of messages $\mra$ with $\mra.\lLOC=\loc$. 
}
\label{fig:RA}
\end{figure}
\textbf{Memory Model Release-Acquire.} As an example of a concrete memory model and its lifting to potentials, we consider the Release-Acquire memory model \RA, following the semantics of~\cite{DBLP:journals/pacmpl/LahavNOPV21}\footnote{We slightly deviate from~\cite{DBLP:journals/pacmpl/LahavNOPV21} as to separate reads and view advancements.}.  
For \RA, we get the following components of the LTS:
\begin{itemize}
    \item The set of memory model internal labels is $\RA.\lTheta\defeq\{\prop\}$. 
    \item Each state $\statera=\tup{\memra,\tsra} \in \RA.\lQ$ consists of a set $\memra$ of {\em messages} and a {\em view} $\tsra$ per thread.  
    A view is a mapping from locations to timestamps (here, $\mathbb{N}$), i.e., $\View \defeq \Loc \rightarrow \mathbb{N}$. We use $\sqcup$ on views for the point-wise maximum on timestamps. Messages $m \in \memra$ 
    record the previously executed writes and store location $\loc$, written value $\val$, and timestamp $\tra$ (denoted $\loc:\val@\tra$) plus the writing thread and the view of the thread when writing, i.e., messages have the type $(\Loc\times\Val\times\N)\times\Tid\times \View$.  The view  $\tsra:\Tid\to\View$ maps each thread to a timestamp for each location.
Thus,   $\RA.\lQ\defeq\mathcal{P}(\Loc\times\Val\times\N\times\Tid\times\View)\times(\Tid\to\View)$. 
 
\item The set of initial states is $\RA.\lQ_0\defeq\set{\tup{\memra_0,\tsra_0}}$. In this, the memory $\memra_0$ contains one message $\tup{\loc:0@0, \ctid{0},V_0}$ where $V_0(x)=0$ for each location $\loc$ and $\tsra_0(\tid)=V_0$ for each thread $\tid$.  
\item The transition relation $\asteplab{}{\RA}$ is given in Fig.~\ref{fig:RA}.
For a given location $\loc$, the read transition chooses the message $\mra\in\memra_{|\loc}$ to read from with a timestamp equal to the thread's view on $\loc$. 
The write transition adds a new message to memory. Its timestamp $\tra$ has to be greater than the current view of the thread on $\loc$, which then gets updated to $\tra$.
The RMW transition simply first executes a read and after that a write transition. To ensure that no thread can write between the written message $\mra_w$ and the one read from, $\mra_r$, the timestamp is set to $\mra_w.\lTS=\mra_r.\lTS +1$.
Finally, the propagate transition updates the view of a thread $\tid$ to the maximum of its current view and the view of a message $\mra$. For this, the timestamp of $\mra$ needs to be greater than the view of $\tid$ on $\mra$'s location. 
\end{itemize}

 \textbf{Lifting of \RA States.}
When lifting \RA states to the potential domain, the auxiliary information \Aux{} we need are  {\em timestamps}, i.e.,  we have $\Aux=\N$.
An \RA state determines multiple potential store lists per thread. 
Each such list $\L\in\DD(\tid)$ starts with the same potential store
\begin{eqnarray*}
   \Delta_\tid(\statera): & \Loc  & \to (\Val \times \Tid \times \N) \times \{\tcovered,\fcovered\} \\
                       & \loc & \mapsto \tup{\mra_\tid.\lVAL,  \mra_\tid.\lTID, \mra_\tid.\lTS ,cov_\statera(\mra_\tid)} 
\end{eqnarray*}
where $\mra_\tid$ is the message in $\statera.\memra_{|\loc}$ with $\mra_\tid.\lTS=\statera.\tsra(\tid)(\loc)$ and 
\begin{align*}
cov_\statera(\mra_\tid)\defeq & \ {\bf if}\ \mra_\tid.\lTS=\max\{\tra\mid \tra=\mra'.\lTS \land \mra'\in\statera.\memra_{|\loc}\}\\
& \lor\exists\mra'\in\statera.\memra_{|\loc}:\mra'.\lTS=\mra_\tid.\lTS+1\ {\bf then}\ \tcovered \ {\bf else}\  \fcovered\ . 
\end{align*} 
Note that this definition is well-defined since for every \RA state there exists exactly one such $\mra_\tid\in\memra$ for each thread $\tid$. 
If $cov_\statera(\mra_\tid)=\tcovered$ for all threads $\tid$, then the message of a new write to $\loc$ automatically receives the highest timestamp with respect to $\loc$ and is therefore the latest one.
This leads us to the lifting map $\map_\RA:\RA.\lQ\to\DD$, where

\begin{align*}
    \map_\RA(\statera)(\tid)\defeq \Delta_\tid(\statera)\cdot \bigcup_{ \statera' \text{ s.t. } \statera\asteplab{\prop}{\RA} \statera' } map_\RA(\statera')(\tid)
\end{align*}

Here, we let $e \cdot \mathsf{L} \defeq \{ \tup{e} \cdot \mathit{L} \mid \mathit{L} \in \mathsf{L}\}$.  Since timestamps are unique for every write to a given location $\loc$ in an \RA-state $\statera$, and propagates have to increase a thread's view, every list in each $\map_\RA(\statera)(\tid)$ ends with the same potential store. 
Hence, $\map_\RA$ is a potential 
according to \cref{def:sra-state}.

\smallskip 
\textbf{Soundness of Rules in \RA.} With the lifting for \RA at hand, we can formally study the soundness of \pico proof rules in \RA. 

\begin{theorem} \label{th:sound-ra}
    All \pico proof rules in Fig.~\ref{fig:rules} are sound for \RA. 
\end{theorem}
All soundness proofs can be found in the appendix. Finally, we have to look at the memory model's internal steps. 

\begin{lemma} \label{th:sound-dist}
    The lifting function $\map_\RA$ guarantees view advancement of $\RA$s internal transitions.
\end{lemma}

As a corollary, we get the following. 

\begin{corollary}
    The liveness proofs of the programs \texttt{\tt{Waiting}} in Fig.~\ref{fig:proof} and \texttt{Ticket} in Fig.~\ref{fig:ticket-proof} are valid under the memory model \RA. 
\end{corollary}

\textbf{Other Memory Models.} 
Besides \RA, all rules of \cref{fig:rules,fig:intrules}  are also sound in memory models $\SC$ and $\TSO$~\cite{DBLP:journals/cacm/SewellSONM10,DBLP:conf/tphol/OwensSS09}. 
The appendix furthermore contains a definition of $\strongC$ and soundness proofs. 
Memory model $\strongC$ does, however, not satisfy rule \textsc{St-Other}\footnote{This rule describes the message passing ability: If a thread sees another thread's write, it will also see that thread's earlier writes.
}. 
In the liveness proof of program \texttt{Waiting}, this rule is used to show that the invariant $\sees{\ctid{2}}{[\gsignal=0];[\glock=1]}$ (which guarantees the order in which $\ctid{2}$ can read something) holds throughout the entire program, especially after executing $\pcl_1$. 
Under $\strongC$, Thread $\ctid{2}$ of \texttt{Waiting} can see the write to $\gsignal$ before it sees the write to $\glock$ and hence might not terminate,  namely when reading $\glock$ to be 0, $\gsignal$ to be 1 and then not seeing any changes to $\gsignal$ anymore. 
Since the liveness proof of the Ticket Lock algorithm does not need the \textsc{St-Other} rule, it is valid for \strongC. 

For \SC and \TSO, on the other hand, both proofs are valid.

\section{Conclusions and Related Work} \label{sec:conclusion} 
\vspace*{-5pt}
In this paper, we have proposed a method for deductively verifying liveness properties of concurrent programs on weak memory models. Our method includes a program logic and associated {\em syntactic} proof rules, thereby making reasoning generic and proofs valid for multiple memory models, namely for all models in which the employed proof rules are sound.

As future work, we intend to look into other memory models, e.g., Partial Store Ordering.  Furthermore, we will study strong fairness (compassion in the terminology of~\cite{DBLP:conf/birthday/MannaP10}) and Manna and Pnueli's rules for response under justice and compassion.


\paragraph{Safety in weak memory.} 
A number of works have proposed program logics and proof calculi for reasoning about concurrent programs on weak memory  with respect to safety properties (e.g.,~\cite{DBLP:conf/icalp/LahavV15,DBLP:conf/ecoop/KaiserDDLV17,ecoop20,DBLP:conf/esop/BilaDLRW22,DBLP:conf/cav/LahavDW23,FestschriftJones,DBLP:conf/fm/BargmannDW24}). Some of them also use \pico~\cite{DBLP:conf/cav/LahavDW23,DBLP:conf/fm/BargmannDW24,FestschriftJones}. For example, \cite{FestschriftJones} have also employed a notion of ``latest value". 

\paragraph{Liveness for SC.} Numerous works have studied general liveness properties or total correctness for concurrent programs wrt.~sequential consistency. 
Owicki and Gries~\cite{DBLP:journals/acta/OwickiG76} were the first to introduce termination proofs within Hoare-like proof calculi, later corrected by  
Apt et al.~\cite{Apt1990}.  Similar proof calculi for total correctness are also studied within rely-guarantee reasoning~\cite{DBLP:journals/toplas/Jones83,DBLP:journals/fac/XuRH97}. Cook et al.~\cite{DBLP:conf/pldi/CookPR07} propose automatic techniques for proving termination by incrementally constructing and eliminating possible interferences by threads. 

These approaches only provide a proof strategy for total correctness where the threads {\em in isolation} terminate (plus their termination proofs fulfill non-interference). Manna and Pnueli~\cite{DBLP:journals/tcs/MannaP91,DBLP:conf/birthday/MannaP10} provide temporal logic proof calculi for general response properties based on fairness. 
Podelski et al.~\cite{DBLP:journals/toplas/PodelskiR07} also employ fairness and provide a technique that constructs abstract transition programs on which fair termination can be checked. Concurrent separation logic was used in an automatic technique by D'Osualdo et al.~\cite{DBLP:journals/toplas/DOsualdoSFG21}.

\paragraph{Liveness in weak memory.} Only a few works have considered progress properties for concurrent programs on weak memory. The only one based on proof calculi is from Colvin et al.~\cite{DBLP:series/lncs/ColvinHHHMS24}. They use rely-guarantee reasoning on weak memory, which, as mentioned above, can only prove thread-wise termination.   
For decidability, Wang et al.~\cite{DBLP:conf/setta/WangPLLL22} prove the liveness properties lock-freedom, wait-freedom, deadlock-freedom, and starvation-freedom to be undecidable for TSO.
Lahav et al.~introduce the notion of memory fairness~\cite{DBLP:journals/pacmpl/LahavNOPV21} and provide a theorem (Theorem 5.3) allowing to reduce the checks for (the absence of) infinite computation to executions in which read statements finally read from writes being maximal wrt.~the modification order in execution graphs.  The tool Dartagnan~\cite{DBLP:conf/tacas/LeonFH020,DBLP:conf/tacas/LeonH021} performs termination checks for axiomatic memory models based on insights of~\cite{DBLP:journals/pacmpl/LahavNOPV21}. In \cite{DBLP:journals/pacmpl/HaasMLG26}, it is extended to prove non-termination for any given fairness assumptions. 
Abdulla et al.~\cite{DBLP:conf/cav/AbdullaAGKV23,DBLP:conf/birthday/AbdullaAGKV24} develop fairness notions and show 
that liveness can be reduced to control state reachability under their notion of fairness.

\bibliographystyle{plainurl}
\bibliography{sample}

\appendix
\section{Program Semantics} 
\label{app:opsem}

In this section, we give the missing rules of the local semantics of \cref{sec:syntax}. We start with a rule that lets us derive local state transitions from the command rules in \cref{fig:pcmdsem}.

\begin{mathpar}
  \inferrule*{
  \pcmd \gg \regstore \astep{\lab_\varepsilon} \regstore' \\\\ \pcmap'=\pcmap [\tau \mapsto \npc(\pc)]
}
{\tup{\pc: \pcmd; \cmd,(\regstore,\pcmap)} \cmdsteppc{\tid}{\lab_\varepsilon}{\pc}\tup{\cmd,(\regstore',\pcmap')}}
\end{mathpar}

Additionally, we need rules to describe if and while statements within the local semantics.

\begin{mathpar}
\inferrule*{
  \regstore(e) = \true \\ \pcmap'=\pcmap[\tid \mapsto \npcT(\pc)]
  }{ \tup{\pc: \ite{\exp}{\cmd_1}{\cmd_2}, (\regstore,\pcmap)} \cmdsteppc{\tid}{\lab_\varepsilon}{\pc} \tup{\cmd_1, (\regstore,\pcmap')}}
  \\
  \inferrule*{\regstore(e) = \false \\ \pcmap'=\pcmap[\tid \mapsto \npcF(\pc)]}{
  \tup{\pc: \ite{\exp}{\cmd_1}{\cmd_2}, (\regstore,\pcmap)} \cmdsteppc{\tid}{\lab_\varepsilon}{\pc} \tup{\cmd_2, (\regstore,\pcmap')}
  } 
\\ 
\inferrule*{
  \regstore(e) = \true \\ \pcmap'=\pcmap[\tid \mapsto \npcT(\pc)]
  }{ \tup{\pc: \while{\exp}{\cmd_1}, (\regstore,\pcmap)} \cmdsteppc{\tid}{\lab_\varepsilon}{\pc} \tup{ \cmd_1, (\regstore,\pcmap')}}
  \\
  \inferrule*{
  \regstore(e) = \false \\ \pcmap'=\pcmap[\tid \mapsto \npcF(\pc)]
  }{ \tup{\pc: \while{\exp}{\cmd_1}, (\regstore,\pcmap)} \cmdsteppc{\tid}{\lab_\varepsilon}{\pc} \tup{\skipc, (\regstore,\pcmap')}}

\end{mathpar}

Here, $\npcT$ and $\npcF$ give the next program position following the true and false branch, respectively,  and $\lab_\varepsilon \in \Lab \cup \{\varepsilon\}$. 
Finally, the last rule lifts local state transitions of {\em sequential} programs to {\em concurrent}  programs. 

\begin{mathpar}
\inferrule*{
  \tup{\cmdmap(\tid),(\regstore,\pcmap)} \cmdsteppc{\tid}{\lab_\varepsilon} {\pc}  \tup{\cmd',(\regstore',\pcmap')}
}
{\tup{\cmdmap,(\regstore,\pcmap)} \cmdsteppc{\tid}{\lab_\varepsilon}{\pc}  \tup{\cmdmap[\tid \mapsto \cmd'],(\regstore',\pcmap')}}
\end{mathpar}

\section{Strong Coherence} \label{app:strongcoh}

In this section, we look at the Strong-Coherence memory model (\strongC). As for \RA, we first restate the operational semantics for \strongC given in \cite{DBLP:journals/pacmpl/LahavNOPV21}.

\noindent \textbf{Operational semantics.}
The semantics of \strongC is very similar to those of \RA. The only difference is in the propagation transition. 
Hence, the set of memory-model internal labels is $\strongC.\lTheta\defeq\RA.\lTheta=\{\prop\}$, the set of states $\strongC.\lQ\defeq\RA.\lQ=\mathcal{P}(\Loc\times\Val\times\N\times\Tid\times\View)\times(\Tid\to\View)$ and the initial state is $\strongC.\lQ_0\defeq\RA.\lQ_0=\tup{\memra_0,\tsra_0}$.

\begin{figure}[t]
\begin{mathpar}
 \inferrule*{ \lab=\prop \\   \mra\in\statera.\memra \\  \exists\tid\in\Tid:\tsra(\tid)(\mra.\lLOC)<\mra.\lTS 
   }
   {\statera\asteplab{\lab}{\strongC} \tup{\lz{\statera.\memra},\lo{\statera.\tsra [\tid \mapsto \statera.\tsra(\tid)[\mra.\lLOC\mapsto\mra.\lTS]]}} } 
\end{mathpar}
\vspace{-10pt}\caption{Propagation transition in the operational semantics of \strongC using colors to highlight the updated \lz{memory} and \lo{view} components 
}
\label{fig:strongC}
\end{figure}

The transition relation $\asteplab{}{\strongC}$ differs from that of \RA only in the propagation transition given in Fig.~\ref{fig:strongC}. In there, the view of a thread could be updated for one location per time, leading to the more relaxed behavior of \strongC compared to \RA.

\noindent \textbf{Potentials of \strongC States.}
Since $\strongC.\lQ=\RA.\lQ$, we can use the same lifting map for $\RA$ and $\strongC$: 
\begin{align*}
    \map_\strongC(\statera)(\tid)\defeq \Delta_\tid(\statera)\cdot \bigcup_{ \statera' \text{ s.t. } \statera\asteplab{\prop}{\strongC} \statera' } \map_\strongC(\statera')(\tid)
\end{align*}
Because of the different propagation transitions, we still get different potentials during the execution of a program. Therefore, not all \pico rules that are sound for \RA still hold for \strongC.
\begin{theorem}\label{th:soundStrongC}
    All \pico proof rules in Fig.~\ref{fig:rules} except for \textsc{St-Other} are sound for \strongC. 
\end{theorem}
The lack of soundness of \textsc{St-Other} is the reason why the proof in \cref{sec:wait} does not apply to \strongC. In order to prove liveness for \strongC with \pico in general, we need the following lemma.
\begin{lemma} \label{th:sound-dist-strongC}
    The lifting function $\map_\strongC$ guarantees view advancement of $\ strongC$'s internal transitions.
\end{lemma}
In the next section, we prove all soundness results for \RA and \strongC.

\section{Soundness Proofs of \RA and \strongC }

Next, we provide the proofs of \cref{th:sound-ra,th:soundStrongC} stating the soundness of the rules in \cref{fig:rules} for memory models \RA and \strongC. 
We start with two rules that always hold for every memory model. 
\begin{lemma}
\label{lem}
    The rules \textsc{Pc} and \textsc{Sub} are sound for every memory model. 
\end{lemma}
For \textsc{Pc} this is the case, since the assertion $\at{\pc}$ only depends on the program semantics, it holds independently of the model, and after executing a transition with program counter $\pc$ in any case the next one is $\npc(\pc)$. For \textsc{Sub}, the transition only depends on the program semantics. It only changes the value of $\regstore$ and the model-dependent potential state stays the same. 
The remaining rules of Fig.~\ref{fig:rules} depend on the respective memory model. 
\smallskip

\subsection{Proof of Theorem~\ref{th:sound-ra}}
For the proof of Theorem~\ref{th:sound-ra}, we show the soundness of the remaining rules of Fig.~\ref{fig:rules} not mentioned in Lemma~\ref{lem} for \RA.

\begin{proof}
Let $\tup{\mathcal{C},(\regstore,\pcmap),\statera}\in\lQ$ with an \RA state $\statera=\tup{\memra,\tsra}$.

\paragraph{{\sc Stbl-St}.} 
Besides updating the program counter, a write transition in \RA adds a new message to $\memra$ and advances the view thread $\tid$ has on $\loc$. This only has an influence on assertions containing $\loc$ or program counters.

\paragraph{{\sc Stbl-Ld}.}
Besides updating $\pcmap$ and the value of $\reg$ in $\regstore$, a read transition in \RA only advances the view of thread $\tid$. For $\loca\in \Loc$, this exclusively has an influence on assertions of the form $\neg(\vmax{\loca}{\tid})$ and the ones containing $\reg$ or program counters. 

\paragraph{{\sc Stbl-Rmw}.}
Since a fetch-and-add transition is only a special case of a composition of a load transition and a store transition, the soundness of \textsc{Stable-RMW} follows from the soundness of the two previous rules.

\paragraph{{\sc St}.}   
 For $\tup{(\regstore,\pcmap),map_\RA(\statera)}\models \vmax{\loc}{\tid}$, we know thread $\tid$ has the highest view on $\loc$. After a new write to $\loc$, that write gets a new, higher timestamp, and $\tid$'s view gets updated. Hence, $\tid$ is still view maximal and can only see the last written value for $\loc$. Therefore, the \textsc{St} rule is sound. 

\paragraph{{\sc St-Other}.}   
Let $\tup{(\regstore,\pcmap),map_\RA(\statera)}\models \sees{\tid}{\inter_\tid}   \wedge  \sees{\tida}{\inter \chop \inter_\tid}$ with $\tid\neq\tida$ and $\loc\notin\fv(\inter_\tid)$. Assume $\Loc_\tid=\fv(\inter_\tid)$. After $\tid$ writes something to $\loc$ we add a new message $m_\loc$ to $\statera.\memra$ with $m.\vra(y)=\tsra(\tid)(y)$ for all $y\in \fv(\inter_\tid)$. Since a propagate of $\tida$ after the write updates the view on all variables, $\tida$ cannot see $m_\loc$ before it sees $\inter_\tid$. Hence, the new state still satisfies $\sees{\tida}{\inter;\inter_\tid}$.

\paragraph{{\sc Rmw}.}   
For the \textsc{Rmw} rule, $\tup{(\regstore,\pcmap),map_\RA(\statera)}\models \cov{\loc}\land \curr{\loc}=\exp_1$ holds. In an \RA state, this means that all messages $\mra\in\memra_{|\loc}$ either have the highest timestamp and $\mra.\lVAL=\exp_1$ or there is a message $\mra'$ with $\mra'.\lTS=\mra.\lTS+1$. An RMW in \RA can then only read from the latest message $\mra$ with $\mra.\lVAL=\exp_1$ and will write a new message $\mra'$ with $\mra'.\lVAL=\exp_1+\exp_2$ and $\mra'.\lTS=\mra.\lTS+1$. The view of $\tid$ for $\loc$ will be updated, and hence, all messages for $\loc$ are covered, and $\tid$ only sees the latest one.

\paragraph{{\sc Ld}.}   
The \textsc{Ld} rule tells us that if $\tid$ only sees the latest message for $\loc$, then it can only read from that message. This is true for \RA, and therefore, this rule is also sound. 

\paragraph{{\sc Ld-Sh}.}   
For the \textsc{Ld-Sh} rule, $\tup{(\regstore,\pcmap),map_\RA(\statera)}\models \sees{\tid}{[\exp(\reg:=\loc)];\inter}$ holds. For \RA, this means $\tid$ currently ether sees $\exp(\reg:=\loc)$ or $\sees{\tid}{\inter}$ holds. After the load, in the first case $e$ holds analogous to rule \textsc{Ld}. Otherwise, $\psi$ holds because of the rule's assumption.

\end{proof}

\subsection{Proof of Theorem~\ref{th:soundStrongC}}
For the proof of Theorem~\ref{th:soundStrongC}, we show the soundness of the remaining rules of Fig.~\ref{fig:rules} not mentioned in Lemma~\ref{lem} for \strongC. \textsc{St-Other} is not sound for \strongC. Note that since the semantics of \RA and \strongC are similar, most of the soundness proofs also are. 

\begin{proof}
Let $\tup{\mathcal{C},(\regstore,\pcmap),\statera}\in\lQ$ with an \strongC state $\statera=\tup{\memra,\tsra}$.

\paragraph{{\sc Stbl-St}.} 
Besides updating the program counter, a write transition in \strongC adds a new message to $\memra$ and advances the view thread $\tid$ has on $\loc$. This only has an influence on assertions containing $\loc$ or program counters.

\paragraph{{\sc Stbl-Ld}.}
Besides updating $\pcmap$ and the value of $\reg$ in $\regstore$, a read transition in \strongC only advances the view of thread $\tid$. For $\loca\in \Loc$, this exclusively has an influence on assertions of the form $\neg(\vmax{\loca}{\tid})$ and the ones containing $\reg$ or program counters. 

\paragraph{{\sc Stbl-Rmw}.}
Since a fetch-and-add transition is only a special case of a composition of a load transition and a store transition, the soundness of \textsc{Stable-RMW} follows from the soundness of the two previous rules.
    
\paragraph{{\sc St}.}   
 For $\tup{(\regstore,\pcmap),map_\strongC(\statera)}\models \vmax{\loc}{\tid}$, we know thread $\tid$ has the highest view on $\loc$. After a new write to $\loc$, that write gets a new, higher timestamp, and $\tid$'s view gets updated. Hence, $\tid$ is still view maximal and can only see the last written value for $\loc$. Therefore, the \textsc{St} rule is sound.

\paragraph{{\sc Rmw}.}   
For the \textsc{Rmw} rule, $\tup{(\regstore,\pcmap),map_\strongC(\statera)}\models \cov{\loc}\land \curr{\loc}=\exp_1$ holds. In a \strongC state, this means that all messages $\mra\in\memra_{|\loc}$, either have the highest timestamp and $\mra.\lVAL=\exp_1$ or there is a message $\mra'$ with $\mra'.\lTS=\mra.\lTS+1$. An RMW in \strongC can then only read from the latest message $\mra$ with $\mra.\lVAL=\exp_1$ and will write a new message $\mra'$ with $\mra'.\lVAL=\exp_1+\exp_2$ and $\mra'.\lTS=\mra.\lTS+1$. The view of $\tid$ for $\loc$ will be updated, and hence, all messages for $\loc$ are covered, and $\tid$ only sees the latest one. 

\paragraph{{\sc Ld}.}   
The \textsc{Ld} rule tells us that if $\tid$ only sees the latest message for $\loc$, then it can only read from that message. This is true for \strongC, and therefore this rule is also sound. 

\paragraph{{\sc Ld-Sh}.}   
For the \textsc{Ld-Sh} rule, $\tup{(\regstore,\pcmap),map_\strongC(\statera)}\models \sees{\tid}{[\exp(\reg:=\loc)];\inter}$ holds. For \strongC, this means $\tid$ currently ether sees $\exp(\reg:=\loc)$ or $\sees{\tid}{\inter}$ holds. After the load, in the first case $e$ holds analogous to rule \textsc{Ld}. Otherwise, $\psi$ holds because of the rule's assumption.

\end{proof}

\subsection{Proofs of Lemma~\ref{th:sound-dist} and Lemma~\ref{th:sound-dist-strongC}}

Next, we prove that $\map_\RA$ and $\map_\strongC$ guarantee view advancement. Because of the similarities between the two memory models, we do this in one proof. 
\begin{proof}
Assume we are in a state $\tup{(\regstore,\pcmap),\map_\M(q)}$ with $\M\in\{\RA,\strongC\}$.
\begin{enumerate}
    \item Let
$$\exists\tid\exists\loc:\dist(\tid,\loc)\neq 0$$
This means that at least one store list in $\DD(\tid)$ has at least two different potential stores. 
Since for both memory models the lifting functions are defined via memory internal steps $\mlab\in \M.\lTheta$, $q\models \En(\mlab)$.
\item \textsc{Stbl-Int}: By definition of $map_\RA(\statera)$, a propagate transition leads to a new state $map_\M(\statera')$ which only differs by all potentials losing some of their first entries. This view advancement can only cause assertions of the type $\neg(\vmax{\loc}{\tid})$ to no longer be valid. \\
\textsc{Dist}: 
Assuming 
$\dist(\tid,\loc)(\map_\M(q))=n>0$ for a thread $\tid$ and a location $\loc$.
For \RA and \strongC, this means the maxima of entries in lists in $\tid$'s potential that are different from their last entry is $n$. 
        After executing the next propagate, depending on the chosen message, each list in $\tid$'s potential loses at least one entry.
        Hence, the maxima of entries in those lists 
        that are different from their last entry is strictly less than $n$ and $\dist(\tid,\loc)(\map_\M(q'))<n$. 
\end{enumerate}

\end{proof}

\end{document}